\documentclass[11pt]{article}
\usepackage{amsmath,amssymb,amsfonts}
\usepackage{amsthm}
\usepackage{graphicx}
\usepackage{color}
\usepackage{subfigure}
\usepackage{fullpage}
\usepackage{float}
\usepackage{mathrsfs}

\newcounter{protoCount}
\newcounter{protoList}
\newsavebox{\tmpbox}
\newlength{\protobox}

\newcommand{\bra}[1]{\ensuremath{\langle #1 |}}
\newcommand{\ket}[1]{\ensuremath{| #1 \rangle}}

\newcommand{\ext}[1]{\langle #1 \rangle}

\newcommand*{\cB}{\mathcal{B}}
\newcommand*{\cC}{\mathcal{C}}
\newcommand*{\cE}{\mathcal{E}}

\newcommand*{\cH}{\mathcal{H}}

\newcommand*{\cJ}{\mathcal{J}}

\newcommand*{\cQ}{\mathcal{Q}}

\newcommand*{\cP}{\mathcal{P}}

\def\01{\{0,1\}}
\newcommand{\hil}{\mathcal{H}}
\newcommand{\id}{\mathbb{I}}
\newcommand{\hf}{\frac{1}{2}}
\newcommand{\bop}{\mathbb{B}}
\newcommand{\Complex}{\mathbb{C}}

\newcommand{\algA}{\mathscr{A}}
\newcommand{\algI}{\mathscr{I}}
\newcommand{\algB}{\mathscr{B}}
\newcommand{\algZ}{\mathscr{Z}}
\newcommand{\setA}{\mathcal{A}}
\newcommand{\setB}{\mathcal{B}}
\newcommand{\setZ}{\mathcal{Z}}
\newcommand{\listM}{\mathcal{M}}
\newcommand{\listL}{\mathcal{L}}
\newcommand{\setS}{\mathcal{S}}

\newcommand{\setR}{\mathcal{R}}

\newcommand{\mP}{\mathcal{P}}
\newcommand{\mA}{\mathcal{A}}
\newcommand{\mB}{\mathcal{B}}
\newcommand{\mS}{\mathcal{S}}

\newcommand{\Yao}{\rm Yao}
\newcommand{\CHSH}{\rm CHSH}

\renewcommand{\exp}[1]{\text{e}^{#1}}

\newenvironment{sdp}[2]{
\smallskip
\begin{center}
\begin{tabular}{ll}
#1 & #2\\
subject to
}
{
\end{tabular}
\end{center}
\smallskip
}

\DeclareMathOperator{\comm}{Comm}

\newcommand{\outp}[2]{|#1\rangle\langle#2|}
\newcommand{\inp}[2]{\langle{#1}|\,{#2}\rangle} 
\newcommand{\isomorph}{\cong}

\DeclareMathOperator{\Tr}{Tr}

\newcommand{\CC}{\mathbb{C}}
\newcommand{\RR}{\mathbb{R}}

\def\one{\leavevmode\hbox{\small1\normalsize\kern-.33em1}}

\newcommand{\ftv}{\omega^f} 
\newcommand{\eps}{\varepsilon}
\newcommand{\aoq}{A_s^a} 
\newcommand{\xoq}{X_{s_j}^{a_j}} 
\newcommand{\txoq}{\hat{X}_{s_j}^{a_j}} 
\newcommand{\hataoq}{\hat A_s^a} 
\newcommand{\tildeaoq}{\tilde A_s^a} 
\newcommand{\boq}{B_t^b}
\newcommand{\hatboq}{\hat B_t^b}
\newcommand{\tildeboq}{\tilde B_t^b}

\newcommand{\ccone}{\cC_\cP}

\newtheorem{theorem}{Theorem}[section]
\newtheorem{lemma}[theorem]{Lemma}

\floatstyle{ruled}
\newfloat{program}{h}{loa}
\floatname{program}{Level $n$ of SDP hierarchy}

\theoremstyle{definition}
\newtheorem{definition}[theorem]{Definition}

\theoremstyle{remark}

\begin{document}

\title{{\bf The quantum moment problem\\ and bounds on entangled multi-prover games}}
\author{Andrew C. Doherty\thanks{School of Physical Sciences, The University of Queensland, Queensland 4072,
Australia}
\and
Yeong-Cherng Liang$^*$
\and
Ben Toner\thanks{
Centrum voor Wiskunde en Informatica, Kruislaan 413, 1098 SJ Amsterdam, The Netherlands}
\and
Stephanie Wehner\thanks{
Institute for Quantum Information, California Institute of Technology,
Pasadena, CA 91125, USA}
}

\maketitle
\thispagestyle{empty}

\begin{abstract}
We study the \emph{quantum moment problem}: Given a conditional
probability distribution together with some polynomial constraints,
does there exist a quantum state $\rho$ and a collection of
measurement operators such that (i) the probability of obtaining a
particular outcome when a particular measurement is performed on
$\rho$ is specified by the conditional probability distribution, and
(ii) the measurement operators satisfy the constraints. For example,
the constraints might specify that some measurement operators must
commute.

We show that if an instance of the quantum moment problem is
unsatisfiable, then there exists a certificate of a particular form
proving this. Our proof is based on a recent result in algebraic
geometry, the noncommutative Positivstellensatz of Helton and
McCullough [{\em Trans. Amer. Math. Soc.}, 356(9):3721, 2004].

A special case of the quantum moment problem is to compute the value
of one-round multi-prover games with entangled provers.  Under the
conjecture that the provers need only share states in
finite-dimensional Hilbert spaces, we prove that a hierarchy of
semidefinite programs similar to the one given by Navascu\'es,
Pironio and Ac\'in [{\em Phys. Rev. Lett.}, 98:010401, 2007]
converges to the entangled value of the game. It follows that the
class of languages recognized by a multi-prover interactive proof
system where the provers share entanglement is recursive.
\end{abstract}

\newpage
\setcounter{page}{1}

\section{Introduction}

The study of multi-prover games has led to many exciting results in
classical complexity theory. A one-round multi-prover \emph{cooperative game of incomplete information} is a game played by a verifier against two provers, Alice and Bob. The strategy of the verifier is fixed. He randomly chooses two questions according to some fixed probability distribution and sends one question to each prover. Alice and Bob then each return an answer to the verifier. The verifier decides whether to accept these answers on the basis of some pre-defined rules of the game that specify whether the given answers are winning answers for the questions sent. To win the game, Alice and Bob may thereby agree on any strategy beforehand, but they may no longer communicate once the game has started. The maximum probability with which Alice and Bob can cause the verifier to accept is known as the \emph{value} of the game. A simple example is the well-known CHSH game~\cite{chsh:nonlocal,cleve:nonlocal}. In this case, the questions and answers are bits. The verifier chooses questions $s \in \01$ and $t \in \01$ uniformly at random and sends $s$ to Alice and $t$ to Bob. In order to win the game, Alice and Bob must reply with bits $a, b \in \01$ such that $s\wedge t = a \oplus b$, i.e., the logical AND of $s$ and $t$ should be equal to the XOR of $a$ and $b$. It is straightforward to verify that the CHSH game has value $3/4$.

Interactive proof systems have received considerable attention since
their introduction by Babai~\cite{babai:ip} and Goldwasser, Micali
and Rackoff~\cite{goldwasser:ip} in 1985.  Of special interest to us
are proof systems with \emph{multiple}
provers~\cite{babai:ipNexp,benor:ip,cai:ip,
feige:twoProverOneRound,feige:mip,lapidot:ip} as introduced by
Ben-Or, Goldwasser, Kilian and Widgerson~\cite{benor:ip}, which can
be described in terms of multi-prover games between a verifier, and
two or more provers. Whereas the provers are computationally
unbounded, the verifier is limited to probabilistic polynomial time.
Both the provers and the verifier have access to a common input
string $x$. The goal of the provers is to convince the verifier that
$x$ belongs to a pre-specified language $L$.  The verifier's aim, on
the other hand, is to determine whether the provers' claim is indeed
valid.  In each round, the verifier sends a poly$(|x|)$ size query to
the provers, who return a polynomial size answer. At the end of the
protocol, the verifier either accepts, meaning that he concludes $x
\in L$, or rejects, based on the messages exchanged and his own
private randomness. A language $L$ has a multi-prover interactive
proof system if there is a protocol such that, if $x \in L$, there
exist answers the provers can give which will cause the verifier to
accept with high probability. However, if $x \not \in L$, then there
exists no strategy for the provers that will only cause the verifier
to accept, except with very low probability.
Here, $x$ and $L$ lead to particular game.
Let MIP\index{MIP} denote the class of languages having a
multi-prover interactive proof system. It has been shown that
classical two-prover interactive proof systems are just as powerful
as proof systems involving more than two provers.
Indeed, Babai, Fortnow and
Lund~\cite{babai:ipNexp}, and Feige and Lov\'asz~\cite{feige:mip}
have shown that a language is in NEXP if and only if it has a
\emph{two}-prover one-round proof
system, i.e., MIP $=$ NEXP.

\subsection{Multi-prover games with entanglement}

In this paper, we study multi-prover games in a quantum setting. In
particular, we allow Alice and Bob to share an entangled quantum
state as part of their strategy. After receiving their questions, the
provers may perform any local measurement on their part of the
entangled state, and decide on an answer based on the outcome of
their measurement. All communication between the verifier and the
provers remains classical. It turns out that sharing entanglement can
increase the probability that the provers can cause the verifier to
accept, an effect known as quantum \emph{nonlocality}. For example,
if the provers share a maximally entangled state of two qubits they
can win the CHSH game (cause the verifier to accept) with probability
$p^*_{\mbox{\tiny CHSH}}\approx 85\% > 3/4$.

We write $\mbox{MIP}^*$ for the set of languages that have
interactive proofs with entangled provers. Very little is known about
MIP$^*$. Most importantly, it was not known before our work whether there exists an algorithm of any complexity for deciding membership in
MIP$^*$, except for extremely restricted classes of games. In particular, if we restrict to games where Alice and Bob each answer a single bit $a,b \in \01$, and the verifier only looks at the XOR of these two bits, then the (entangled) value of the game can be computed in time polynomial in the number of questions~\cite{tsirel:separated,cleve:nonlocal}.  Let
$\oplus\mbox{MIP}[2]$\index{MIP[2]} denote the restricted class where
the verifier's output is a function of the XOR of two binary answers.
Then $\oplus\mbox{MIP}^*[2] \subseteq \mbox{EXP}$~\cite{cleve:nonlocal}, while it is known that
$\oplus\mbox{MIP} = \mbox{NEXP}$, for certain completeness and
soundness parameters~\cite{Hastad:502098}, i.e., the resulting
proof system is significantly weakened if the provers are allowed to
share entanglement. In fact, such a proof system can even be
simulated using just a single quantum prover, i.e.,
$\oplus\mbox{MIP}^*[2] \subseteq \mbox{QIP}(2) \subseteq
\mbox{EXP}$~\cite{wehner05c,kitaev&watrous:qip}.

Unfortunately, very little is known for more general games where we allow shared entanglement between the provers, and where Alice and Bob give much longer answers, or when we have more than two provers. Unlike in the classical case, general multi-prover games are not equivalent to two-prover games when the provers share entanglement~\cite{julia:mip}. For two-prover unique games (where for each pair of questions there exists exactly one pair of winning answers), it is known that we can approximate the value of a two-prover game to within a certain accuracy in polynomial time~\cite{oded:unique}. Masanes~\cite{masanes:nonlocal} has shown how to compute the value of multi-prover games where the questions to, and the answers from each prover are bits. But even for very small games with a very limited number of questions, the entangled value is typically unknown~\cite{massar:tsirel}.

Assuming that the provers share quantum entanglement is a reasonable
model because it captures the properties of a multi-prover game that
a verifier can enforce physically: while the verifier can enforce the
condition that the provers cannot communicate by ensuring that they are
spacelike-separated, he has no way to ensure that provers in a
quantum universe do not share entanglement. Multi-prover games with
entangled provers are also known as \emph{non-local games} with
entanglement. The $\emph{entangled value}$ of a game is the maximum
probability with which Alice and Bob can win using entanglement.
Here, we are concerned with the following question: How can we
compute the entangled value of non-local games with multiple provers?
And, how can we decide membership of MIP$^*$?

\subsection{Results}

\paragraph{Quantum moment problem.}

To reach our goal, we first introduce the \emph{quantum moment
problem} that is a generalization of our problem. Informally, the
quantum moment problems asks whether given a conditional probability
distributions and some polynomial constraints on observables, we can
find a quantum state and quantum measurements satisfying the
constraints, that provide us with the required probabilities. We may
use the constraints to impose certain restrictions on the form of our
quantum measurements. For example, we may want to demand that two
measurement operators act on two or more distant subsystems
independently. Determining whether there is an entangled strategy for
a multi-prover game that achieves a certain winning probability is a
special case of the quantum moment problem.

Other special cases of the quantum moment problem include the
so-called classical marginal problem~\cite{pitowsky91:_correl}, which
asks whether given certain marginal distributions, we can find a
joint distribution that has the desired marginals. Our problem is
also closely related to the \emph{quantum} marginal problem in which
the aim is to find a density matrix for a multipartite quantum system
that is consistent with a specified set of reduced density matrices
for specific subsystems. This problem is QMA-complete and has
attracted a lot of interest recently~\cite{matthias:qma}. A
special case is $N$-representability, an important problem with a
long history in quantum chemistry~\cite{klyachko}.
One key difference between our quantum moment problem and the quantum
marginal problem is that in the latter case  the dimension of the
quantum state, and its various subsystems, is specified. In the
quantum moment problem the aim is to find a state satisfying the
given constraints in a quantum system of any, possibly infinite,
dimension. Finally, it may also be possible to treat games with
\emph{quantum} verifier within the framework of the quantum moment
problem.

\paragraph{Refuting unsatisfiable instances.}
We describe a general way of proving that an instance of a quantum
moment problem is unsatisfiable. The proof follows from a recent result of
Helton and McCullough~\cite{helton04:_posit_for_non_commut_polyn}, a
Positivstellensatz for polynomials in noncommuting variables. The
choice of polynomials will define a particular instance of the
quantum moment problem, where the variables correspond to measurement
operators. In Helton and McCullough's result the noncommuting
variables are required to satisfy certain polynomial equality and
inequality constraints but can be evaluated in any quantum system,
even one of infinite dimensions. Informally, the Positivstellensatz states that any
such polynomial that is positive can be written a sum of squares, a
form that makes it obvious that the polynomial is positive. By
positive we mean that, whenever the constraints are satisfied, the polynomial is positive semidefinite, i.e. it has a positive expectation value
for all quantum states in whatever quantum system.
Such a representation as a sum of squares acts as a
certificate for the unsatisfiability of an instance of a quantum
moment problem. Certificates of this kind have often been used in the
theoretical physics literature to place very general bounds on
quantum correlations (see for
example~\cite{glauber1963a}).
Helton and McCullough's result shows that
such certificates are all that is ever required to demonstrate that
an instance of the quantum moment problem is unsatisfiable.

\paragraph{Tensor products and commutation.}
In order to apply the Positivstellensatz to obtain bounds on the
entangled values of two-prover games, we need to incorporate a
constraint in the corresponding quantum moment problem that ensures
Alice and Bob's measurements act on different subsystems $\hil_A$ and $\hil_B$. When Alice
and Bob share a quantum system of some finite dimension, this means
that one demands that the Hilbert space $\cH$ describing this system
decomposes as $\hil = \hil_A \otimes \hil_B$. Alice and Bob's
measurements should be of the form $A_s \otimes B_t$ with $A_s \in \hil_A$
and $B_t \in \hil_B$ for all questions $s$ and $t$. Unfortunately, we
can only apply the Positivstellensatz when the constraints are
polynomials in $A_s$ and $B_t$. Thus we need an additional trick to impose an
explicit tensor product structure.
To get around this problem, we
demand that for all $s$ and $t$ we have $[A_s,B_t]=0$, i.e.,
that all measurement operators of Alice commute with all those of
Bob. If the Hilbert space is finite-dimensional, then imposing the
commutativity constraints is actually \emph{equivalent} to demanding
a tensor product structure. This result is well-known in the
mathematical physics literature~\cite{summers:qftIndep}. Here, we
 provide a simple version of this argument accessible
from a computer science perspective, which directly applies to our
analysis of multi-prover games.

From a physics point of view, however, the usual requirement on
observables that can be measured in space-like separated regions is
that they should \emph{commute}, not that they should have a tensor
product factorization. Indeed, this commutativity requirement on
local observables is regarded by many as an axiom that should be
satisfied by any reasonable quantum mechanical theory of
nature~\cite{haag1996a}. Unfortunately, when the algebra of
observables cannot be represented on a finite-dimensional Hilbert
space it is an open question whether this commutativity property
implies the existence of a tensor product factorization. Our results
will provide bounds on the values of multi-prover games that are
valid for all quantum systems whenever the observables of different
players commute. We will refer to
the maximum probability of winning a game $G$ with (possibly)
infinite-dimensional operators satisfying commutativity constraints
as the
\emph{field-theoretic value} $\ftv(G)$ of the game.
It is an open question whether this is the same as
the usual entangled value of the game. Since MIP$^*$ was defined with
a tensor product structure in mind, we here define the class MIP$^f$,
where the tensor products are replaced by the commutative requirement. The
class MIP$^f$ seems more appropriate to our main motivation of
studying the power of multi-prover games where the provers are only
limited by what they can achieve physically. Restricting Alice and
Bob to sharing finite-dimensional systems does not seem natural from
a physical perspective.

\paragraph{A hierarchy of semidefinite programs.}
We know that the Positivstellensatz leads to certificates that tell
us when a particular quantum moment problem is unsatisfiable. But how
can we find such certificates? If we place a bound on the size of the
certificate, then the problem of determining whether there exists a
certificate of that size can be formulated as a semidefinite program
(SDP)~\cite{L.Vandenberghe:SR:1996,Boyd:04}.  In particular,
searching for certificates of increasing size yields a hierarchy of
SDPs. The resulting hierarchy is very similar to the one presented in
a groundbreaking paper of Navascu\'es, Pironio and
Ac\'in~\cite{navascues:_bound}, which partly motivated this work.

In many applications, including multi-prover games, we are not only
interested in whether a specific instance of the moment problem is
satisfiable but in finding the best possible bound on some linear
combination of moments. Once again fixing the size of the
certificates of infeasibility straightforwardly leads to a hierarchy
of SDPs that provide progressively tighter bounds. For multi-prover
games $G$, it was previously not known whether the solutions to this
hierarchy of SDPs converged to the entangled value of the game, which
we denote by $\omega^*(G)$. Here we \emph{almost} show this. What we
actually show is that the hierarchy converges to the
field-theoretic value (see above) of a non-local game $G$.
In the language of the quantum moment problem, we wish to know if there exists an entangled strategy for $G$
such that the provers win with some fixed probability $p$.
However, the
Positivstellensatz only yields a certificate that there is no
entangled strategy that wins with probability $p$ if there is also no
such strategy even \emph{with infinite-dimensional
  measurement operators}. If the measurement operators are
infinite-dimensional, then the commutativity constraints do not
necessarily imply the existence of a tensor product
structure.
In other words, we show that our hierarchy converges to the field-theoretic value of the game.

\paragraph{MIP$^*$ is recursive}

Since our hierarchy converges, we can compute the value of an
entangled game and hence obtain an algorithm for deciding membership
of MIP$^*$ (under the assumption that the optimal value is achieved
with finite-dimensional operators) and of MIP$^f$. This implies that
these classes are recursive.

\paragraph{Examples: The $I_{3322}$ inequality and Yao's inequality.}
Finally, we demonstrate the power of our technique by providing an
extremely simple, algorithmically constructed, certificate bounding
the value of a two-party Bell inequality~\cite{Bell:64a} known as the $I_{3322}$ inequality~\cite{D.Collins:JPA:2004}, and a multi-player non-local game suggested by Yao and collaborators~\cite{yao:inequality}.

\subsection{Open Questions}
With respect to the above discussion, it would be interesting to
know whether there are games $G$ such that $\omega^*(G)$ is strictly
less than $\ftv(G)$. Can it really help the provers to have
infinite-dimensional systems when the number of questions and answers
in the game are finite? One way to establish that there is no
advantage to having infinite-dimensional systems would be to `round'
the SDP hierarchy directly to a quantum strategy with finite
entanglement, bypassing the (nonconstructive) Positivstellensatz
altogether. For XOR-games, the first level of our hierarchy is tight
and it is well-known how a solution of the SDP can be transformed
into a quantum strategy via so-called Tsirelson's vector
construction~\cite{tsirel:original, tsirel:separated,tsirel:hadron}.
However, there exist many non-local games, for which the first level
of the hierarchy does not provide us with the optimal value of the
game, but merely gives us an upper bound.  This fact alone shows that
for general games, we cannot find such a nice embedding of vectors
into observables as can be done for XOR-games. However, something
similar may still be possible for restricted classes of games, exhibiting
a likewise special structure.

We also do not establish anything about the rate of convergence of
the SDP hierarchy. In some numerical experiments with small games,
the low levels of the SDP hierarchy do yield optimal solutions.
Establishing this in general would provide an upper bound on
$\mbox{MIP}^*$. We have made partial progress on this question by
proving convergence for a particular hierarchy of SDPs.

\subsection{Related work}
In Ref.~\cite{navascues:_bound}, a paper which partly inspired this
work, Navascu\'es, Pironio, and Ac{\'\i}n (NPA), defined a closely
related semidefinite programming hierarchy.  Subsequently, and
independently of us, NPA have proved that their semidefinite
programming hierarchy converges to the field-theoretic value of the
game~\cite{navascues08:long}. Our paper and theirs are complementary:
While our work emphasizes the connection with Positivstellensatz of
Helton and McCullough, NPA prove convergence directly. Their proof
has a number of advantages: most notably, when their hierarchy
converges to the field-theoretic value of the game at a finite level,
NPA obtain a bound on the dimension of the state required to
reproduce the correlations.
NPA have also shown that their new technique for proving
convergence can be extended to general polynomial optimization
problems in noncommutative variables~\cite{navascues08:general}.

Finally, our techniques have recently extended by
Ito, Hirotada, and Matsumoto to the case of games with quantum messages between
verifier and provers~\cite{Ito08}.

\subsection{Outline}

In Section~\ref{prelim}, we provide an introduction to non-local
games including all necessary definitions. Section~\ref{qmpDef} then
defines the quantum moment problem, and Section~\ref{section:tools} introduces our main tools.
In particular, Section~\ref{section:commutation} provides an explanation of why
we obtain a tensor product structure from commutation relations, and in
Section~\ref{section:satz}, we show that if a quantum moment problem
is unsatisfiable, we can find certificates of this fact using the
Positivstellensatz. We then use these tools in our SDP hierarchy in
Section~\ref{findingBounds} and conclude in
Section~\ref{Sec:Examples} with some explicit examples.

\section{Preliminaries}\label{prelim}

\subsection{Notation}

We assume general familiarity with the quantum
model~\cite{nielsen&chuang:qc}. In the following, we use $A^\dagger$
to denote the conjugate transpose of a matrix $A$.
A matrix is \emph{Hermitian} if and only if $A^\dagger = A$. We write
$A \geq 0$ to indicate that a matrix $A$ is \emph{positive
semidefinite}, i.e., it is Hermitian and has no negative eigenvalues.
We also use $A =
0$ to express that $A$ is the all-zero matrix and $A \neq 0$ to
indicate that $A$ has at least one non-zero entry. The $(i,j)$--entry
of $A$ will be denoted by $[A]_{i,j}$. For two matrices $A$ and $B$
we write their commutator as $[A,B] = AB - BA$. We use $\hil$ to
denote a Hilbert space and $\hil_{k}$ the Hilbert space belonging to
subsystem $k$. $\id_{k}$ is the identity on system $k$, and
$\bop(\hil)$ denotes the set of all bounded operators on the Hilbert
space $\hil$. Unless stated otherwise, we take all systems to be
finite-dimensional. We will also employ the shorthand
$$
\bop(\hil)^{\times n} := \underbrace{\bop(\hil) \times \ldots \times \bop(\hil)}_{n}
$$
for the $n$-fold Cartesian product of $\bop(\hil)$,
and let $[n] := \{1,\ldots,n\}$. Furthermore, we will use $|\setS|$
and $|\listL|$ to denote the number of elements of a set $\setS$ and
list $\listL$ respectively.

For the purpose of subsequent discussion, we now note that a
Hermitian polynomial $p(X)$ in noncommutative variables $X =
(X_1,\ldots,X_k)$ is a \emph{sum of squares} (SOS) if there exist
polynomials (matrices) $r_j$ of appropriate dimension such that
$$
p(X) = \sum_j r_j^\dagger r_j.
$$
It is important to note that if $p(X)$ is an SOS polynomial, it is
also a positive semidefinite matrix, i.e., $p(X) \geq 0$.

\subsection{Games}

As an example application of the quantum moment
problem, we will consider cooperative games among $n$ parties. For
simplicity, we first describe the setting for only two parties,
henceforth called Alice and Bob. A generalization is straightforward.
Let $S$, $T$, $A$ and $B$ be finite sets, and $\pi$ a probability
distribution on $S \times T$. Let $V$ be a predicate on $S \times T
\times A \times B$. Then $G = G(V,\pi)$ is the following two-person
cooperative game: A pair of questions $(s,t) \in S \times T$ is
chosen at random according to the probability distribution $\pi$.
Then $s$ is sent to Alice, and $t$ to Bob. Upon receiving $s$, Alice
has to reply with an answer $a \in A$. Likewise, Bob has to reply to
question $t$ with an answer $b \in B$. They win if $V(s,t,a,b) = 1$
and lose otherwise. Alice and Bob may agree on any kind of strategy
beforehand, but they are no longer allowed to communicate once they
have received questions $s$ and $t$. The \emph{value} $\omega(G)$ of
a game $G$ is the maximum probability that Alice and Bob win the
game. In what follows, we will write $V(a,b|s,t)$ instead of
$V(s,t,a,b)$ to emphasize the fact that $a$ and $b$ are answers given
questions $s$ and $t$.

Here, we are particularly interested in non-local games and where
Alice and Bob are allowed to share an arbitrary entangled state
$\ket{\psi}$ to help them win the game. Let $\hil_A$ and $\hil_B$
denote the Hilbert spaces of Alice and Bob respectively. The state
$\ket{\psi} \in \hil_A \otimes \hil_B$ is part of the quantum
strategy that Alice and Bob can agree on beforehand. This means that
for each game, Alice and Bob can choose a specific $\ket{\psi}$ to
maximize their chance of success. In addition, Alice and Bob can
agree on quantum measurements where we may without loss of generality
assume that these are projective
measurements~\cite{cleve:nonlocal}.\footnote{By Neumark's theorem,
any generalized measurements described by positive-operator-valued measure can be implemented as projective
measurements in some higher dimensional Hilbert space. See, for
example, pp. 285 of Ref.~\cite{A.Peres:Book:1995}.} For each $s \in
S$, Alice has a projective measurement described by $ \{A_s^a: a \in
A\} $ on $\hil_A$. For each $t \in T$, Bob has a projective
measurement described by $ \{B_t^b: b \in B\} $ on $\hil_B$. For
questions $(s,t) \in S \times T$, Alice performs the measurement
corresponding to $s$ on her part of $\ket{\psi}$ which gives her
outcome $a$. Likewise, Bob performs the measurement corresponding to
$t$ on his part of $\ket{\psi}$ with outcome $b$. Both send their
outcome, $a$ and $b$, back to the verifier. The probability that
Alice and Bob answer $(a,b) \in A \times B$ is then given by
$$
\bra{\psi}A_s^a \otimes B_t^b\ket{\psi}.
$$
We can now define:

\begin{definition}\label{definition:1}
The \emph{entangled value} of a two-prover game with classical
verifier $G = G(\pi, V)$ is given by:
\begin{align}\label{Eq:omega*}
    \omega^*(G) =& \lim_{d \to \infty}
    \max_{\stackrel{\ket \psi \in \CC^d \otimes \CC^d}{\|\,\ket\psi\,\| =1}}
\max_{\aoq, \boq}
    \sum_{a,b,s,t} \pi(s,t) V(a,b|s,t) \bra \psi \aoq \otimes \boq \ket \psi,
\end{align}
where $\aoq \in \bop(\hil_A)$ and $\boq \in \bop(\hil_B)$ for some
Hilbert space $\hil = \hil_A \otimes \hil_B$,
satisfying $\aoq,\boq \geq 0$, $\sum_a \aoq = \id_A$, $\sum_b
\boq = \id_B$ for all $s \in S$ and $t \in T$.
\end{definition}

We also define a more general version of this statement, which will
provide an upper bound to the quantum value of the game:

\begin{definition}
\label{Dfn:ftv} The \emph{field-theoretic value} of a two-prover game
with classical verifier $G = G(\pi, V)$ is given by:
\begin{align}
    \ftv(G) =& \sup_{\aoq, \boq}
\Bigl\|
\sum_{a,b,s,t} \pi(s,t) V(a,b|s,t)
    \aoq \boq \Bigr\|,
\end{align}
where $\|O\|$ is the operator norm of $O$, $\aoq \in \bop(\hil)$ and $\boq\in \bop(\hil)$ for some Hilbert
space $\hil$,
satisfying $\aoq, \boq \geq 0$,
$\sum_a \aoq = \sum_b \boq = \id$ for all $s,t$, and $[\aoq, \boq ] =
0$ for all $s\in S$, $t \in T$, $a \in A$, and $b \in B$.
\end{definition}

\begin{lemma}
  \label{lemma:1}
    Let $G = G(\pi, V)$ be a two-prover game with classical verifier.
    Then $\omega^*(G) \leq\ftv(G)$.
\end{lemma}
\begin{proof}
Let $\eps > 0$. Choose $d$ sufficiently large so that there is a
normalized state $\ket \psi$ and operators $\aoq$, $\boq$ defining a
strategy with winning probability at least $\omega^*(G) - \eps$. Let
$\hataoq = \aoq \otimes \id_B$ and $\hatboq = \id_A \otimes \boq$.
Then $\hataoq$ and $\hatboq$ are positive semidefinite operators on
$\CC^{d^2}$ satisfying all the conditions in
Definition~\ref{Dfn:ftv}. Finally,
\begin{align*}
    \ftv(G) =& \sup_{\tilde{\aoq}, \tilde{\boq}} \Bigl\|\sum_{a,b,s,t} \pi(s,t) V(a,b|s,t)
    \tilde{\aoq} \tilde{\boq} \Bigr\|\\
    \geq&\, \Bigl\|\sum_{a,b,s,t} \pi(s,t) V(a,b|s,t)
    \hataoq \hatboq \Bigr\|\\
    \geq&\, \bra \psi \Bigl( \sum_{a,b,s,t} \pi(s,t) V(a,b|s,t)
    \hataoq \hatboq \Bigr) \ket \psi\\
    \geq&\, \omega^*(G) - \eps.
\end{align*}
Since $\eps$ was arbitrary, the result follows.
\end{proof}

In our examples, we will sometimes use the term \emph{Bell inequality}\cite{Bell:64a} to refer to a particular non-local game. This is an equivalent formulation, where we only consider terms of the form $\bra{\psi}\aoq\boq\ket{\psi}$. The value of the game can then be
obtained by averaging. In inequalities where Alice and Bob have, respectively, two measurement outcomes for each possible choice of measurement setting (i.e., $A = B = \01$), their measurements can be described by observables of the form $A_s = A_s^0 - A_s^1$ and $B_t = B_t^0 - B_t^1$ respectively. In this case, we state inequalities in the form of the observables $A_s$ and $B_t$ where we will use the shorthand $\langle A_sB_s \rangle = \bra{\psi}A_s B_s\ket{\psi}$.

Note that it is straightforward to extend the above definitions to
the setting involving multiple players, but the resulting terms will
be much harder to read. When considering games among $N$ players
$P_1,\ldots,P_N$, let $S_1,\ldots,S_N$ and $A_1,\ldots,A_N$ be finite
sets corresponding to the possible questions and answers
respectively. Let $\pi$ be a probability distribution on $S_1 \times
\ldots \times S_N$, and let $V$ be a predicate on $A_1 \times \ldots
\times A_N \times S_1 \times \ldots \times S_N$. Then $G = G(V,\pi)$
is the following $N$-player cooperative game: A set of questions
$(s_1,\ldots,s_N) \in S_1 \times \ldots \times S_N$ is chosen at
random according to the probability distribution $\pi$. Player $P_j$
receives question $s_j$, and then responds with an answer $a_j \in A_j$. The players win if and only if $V(a_1,\ldots,a_N|s_1,\ldots,s_N) = 1$. Let $\ket{\psi}$ denote the players' choice of state, and let
$X_j := \{X_{s_j}^{a_j}\mid a_j \in A_j\}$ denote the positive-operator-valued measure(ment) (POVM) of player $P_j$ for question $s_j \in S_j$, i.e., $\sum_{a_j} X_{s_j}^{a_j}=\id_j$ and
$X_{s_j}^{a_j}\ge 0$ for all $a_j$. The value of the game can now be
written as
\begin{equation*}
    \omega^*(G) =\lim_{d\to\infty}
    \max_{\stackrel{\ket \psi \in \left(\CC^d\right)^{\otimes N}}{\|\,\ket\psi\,\| =1}}
    \max_{X_1,\ldots,X_{N}} \!\!\sum_{s_1,\ldots,s_N}\!\!\!\pi(s_1,\ldots,s_N)\!\!
    \sum_{a_1,\ldots,a_N}\!\!\!\!V(a_1,\ldots,a_N|s_1,\ldots,s_N)
    \bra{\psi} X_{s_1}^{a_1} \otimes \ldots \otimes X_{s_N}^{a_N} \ket{\psi},
\end{equation*}
where the maximization is taken over all legitimate POVMs $X_j$ for
all $j \in [N]$. Similarly, we can now write the field-theoretic
value of the game as
\begin{equation*}
    \ftv(G) = \sup_{X_1,\ldots,X_N} \Bigl\|
    \sum_{s_1,\ldots,s_N}\!\!\! \pi(s_1,\ldots,s_N) \sum_{a_1,\ldots,a_N}\!\!\!
    V(a_1,\ldots,a_N|s_1,\ldots,s_N)\,X_{s_1}^{a_1}\ldots X_{s_N}^{a_N}
    \bigr\|,
\end{equation*}
where we now have $\sum_{a_j} X_{s_j}^{a_j}=\id$ for all $a_j,s_j,j$
and $[X_{s_j}^{a_j},X_{s_{j'}}^{a_{j'}}]=0$ for all $j\neq j'$.

\subsection{Proof systems}

Interactive proof systems can be phrased as such games. For
completeness, we here provide a definition of MIP. We refer to the
introduction and the previous section for an explanation of the notions of a tensor
product form vs. commutation relations.

\begin{definition}
For $0 \leq s < c \leq 1$, let $\oplus\mbox{MIP}^*_{c,s}[k]$ denote
the class of all languages $L$ recognized by a classical $k$-prover
interactive proof system with entanglement such that:
\begin{itemize}
\item The interaction between the verifier and the provers is
    limited to one round and classical communication. The
    verifier chooses $k$ questions from a finite set of possible
    questions, according to a fixed probability distribution
    known to the provers, and sends one question to each prover. Afterwards,
the provers may perform any measurement that
    has tensor product form on a shared state $\ket{\psi}$ that
    they have chosen ahead of time. Each prover returns an answer
    to the verifier, whose decision function is known to the
    provers.
\item If $x \in L$ then there exists a strategy for the provers
    for which the probability that the verifier accepts is at
    least $c$ (the \emph{completeness} parameter).
\item If $x \notin L$ then, whatever strategy the $k$ provers
    follow, the probability that the verifier accepts is at most
    $s$ (the \emph{soundness} parameter).
\end{itemize}
\end{definition}

\begin{definition}
For $0 \leq s < c \leq 1$, let $\oplus\mbox{MIP}^f_{c,s}[k]$ denote
the class corresponding to a modified version of the previous
definition: here we merely ask that the measurements operators
between the different players commute.
\end{definition}

\section{The quantum moment problem}\label{qmpDef}

\subsection{General form}
Let us now state the quantum moment problem in its most general form,
before explaining its connection to non-local games. Intuitively, the
quantum moment problem states that given a certain probability
distribution, is it possible to find quantum measurements and a state
that provide us with such a distribution?
\begin{definition}[\textbf{Quantum moment problem}]
Given a list of numbers $\listM = (m_i \mid m_i \in [0,1])$, a set of
polynomial equations $\setR = \{r = 0\mid r: \bop(\hil)^{\times
|\listM|} \rightarrow \bop(\hil)\}$, and polynomial inequalities
$\setS = \{s \geq 0 \mid s: \bop(\hil)^{\times |\listM|} \rightarrow
\bop(\hil)\}$, does there exist said Hilbert space $\hil$, operators
$M_i \in \bop(\hil)$ and a state $\rho \in \bop(\hil)$ such that
\begin{enumerate}
\item For all $m_i \in \listM$, $\Tr(M_i \rho) = m_i$.
\item For all $r \in \setR$, $r(M_1,\ldots,M_{|\listM|}) = 0$.
\item For all $s \in \setS$, $s(M_1,\ldots,M_{|\listM|}) \geq 0$.
\end{enumerate}
\end{definition}

\subsection{Non-local games}\label{qmp:nonlocal}
In this paper, we are particularly interested in a special case of
the quantum moment problem, where we consider measurements on many
space-like separated systems as in the setting of non-local games.
For simplicity, we will explain the connection to non-local games for
only two such systems, Alice $\hil_A$ and Bob $\hil_B$, where it is
straightforward to extend our arguments to more than two. On each
system $\hil_A$ and $\hil_B$, we want to perform a finite set of
possible measurements $S$ and $T$ each of which has the same finite
set of outcomes $A$ and $B$ respectively. Let $m^A(a|s)$ and
$m^B(b|t)$ denote the probability that on systems $\hil_A$ and
$\hil_B$ we obtain outcomes $a \in A$ and $b \in B$ given measurement
settings $s \in S$ and $t \in T$ respectively. Furthermore, let
$m^{AB}(a|s,b|t)$ denote the joint probability of obtaining outcomes
$a$ and $b$ given settings $s$ and $t$ when performing measurements
on systems $\hil_A$ and $\hil_B$.

Informally, our question is now: Given probabilities
$m^{AB}(a|s,b|t)$, does there exist a shared state $\rho$ such that
we can find measurements on the individual systems $\hil_A$ and $\hil_B$
that lead to such probabilities? Let's first consider what polynomial
equations and inequalities we need to express our problem in the
above form. First of all, how can we express the fact that we want
our measurement operators to act on the individual systems $\hil_A$ and
$\hil_B$ alone? I.e, how can we ensure that the measurement operators
have tensor product form? We will show in Lemma~\ref{tensorProduct}
that we are guaranteed to observe such a tensor product form if and
only if for all $s \in S$, $a \in A$, $t \in T$ and $b \in B$ we have
$[A_s^a,B_t^b] = 0$, where we used $A_s^a$ and $B_t^b$ to denote the
measurement operators of Alice and Bob corresponding to measurement
settings $s$ and $t$ and outcomes $a$ and $b$ respectively. Hence, we
need to impose the polynomial equality constraints of the form
$[A_s^a,B_t^b] = 0$.

Furthermore, we want that for both systems $\mA$ and $\mB$, we obtain
a valid measurement for each measurement setting $s \in S$ and $t \in
T$. I.e., we impose further polynomial equality constraints for all
$s \in S$ and $t \in T$ of the form
$$
\sum_{a \in A} \aoq - \id = 0
\mbox{ and }
\sum_{b \in B} \boq - \id = 0,
$$
and finally the following polynomial inequality constraints
for all $a \in A$ and $b \in B$
$$
\aoq \geq 0 \mbox{ and }
\boq \geq 0.
$$
Recall, that we may restrict ourselves to considering projective measurements.
We may thus add the equality constraints
$$
(\aoq)^2 = \aoq \mbox{ and } (\boq)^2 = \boq,
$$
which automatically imply that $\aoq, \boq \geq 0$. For simplicity, we will later
use this constraint instead of the previous one.

In this paper, we are mainly concerned with the (weighted) average of
the probabilities of generating certain outcomes in a non-local game.
In other words, we wanted to know if there exist operators of the
above form such that
$$
\nu =  \Bigl\|\sum_{a,b,s,t} \pi(s,t) V(a,b|s,t) \aoq \boq\Bigr\|
$$
for some success probability $\nu$. Semidefinite programming will
allow us to turn the question of existence into an optimization
problem.

\section{Tools}\label{section:tools}

For our analysis we first need to introduce two key tools. The first
one allows us to deal with the fact that we want measurements to have
tensor product form. Our second tool is an extension of the
non-commutative Positivstellensatz of Helton and McCullough to the
field of complex numbers, from which we will derive a converging
hierarchy of semidefinite programs.

\subsection{Tensor product structure from commutation relations}\label{section:commutation}

We now first show that imposing commutativity constraints does indeed
give us the tensor product structure required for our analysis of
non-local games. It is well-known that the following statement holds
within the framework of quantum
mechanics~\cite{summers:qftIndep}.\footnote{an algebra of type-I}

In Appendix~\ref{Sec:Tensor}, we provide a simple version of this
argument accessible from a computer science perspective, which
directly applies to the task at hand.
\begin{lemma}\label{tensorProduct}
Let $\hil$ be a finite-dimensional Hilbert space, and let $\{X_{s_j}^{a_j} \in
\bop(\hil)\mid \mbox{ for all } j \in [N] \mbox{ and for all } s_j \in S_j, a_j \in A_j\}$.
Then the following two statements are equivalent:
\begin{enumerate}
\item For all $j,j' \in [N]$, $j\neq j'$, and all $s_j \in S_j$, $s_{j'} \in S_{j'}$, $a_j \in A_j$ and $a_{j'} \in {A_{j'}}$ it holds that
$[X_{s_j}^{a_j},X_{s_{j'}}^{a_{j'}}] = 0$.
\item There exist Hilbert spaces $\hil_1, \ldots,\hil_N$ such that $\hil
    = \hil_1 \otimes \ldots \otimes \hil_N$ and for all $j \in [N]$, all $s_j \in S_j$, $a_j \in A_j$ we
    have $X_{s_j}^{a_j} \in \bop(\hil_j)$.
\end{enumerate}
\end{lemma}

\subsection{Positivstellensatz}\label{section:satz}

Our second tool, the Positivstellensatz (in combination with
semidefinite programming) will allow us to find certificates for the
fact a quantum moment problem is infeasible. For simplicity, we here
describe the Positivstellensatz from the perspective of non-local
games. An extension to the general quantum moment problem is possible
and will be provided in a longer version of this paper.
Our results
follow almost directly from Helton and McCullough's work and our
proof closely follows that in
Ref.~\cite{helton04:_posit_for_non_commut_polyn}.  We have chosen to
provide a complete proof of the Positivstellensatz for three reasons:
(i) the proof is more straightforward in our concrete setting, (ii)
Helton and McCullough's theorem is formulated for symmetric operators
over the field $\RR$, and we need to work with Hermitian operators
over the field $\CC$, and (iii) so we can highlight the
nonconstructive steps in the proof. We first define:

\begin{definition}[Convex Cone $\ccone$]
  \label{definition:3}
    Let $\cP$ be a collection of Hermitian  polynomials in
    (noncommutative) variables $\{X_{s_j}^{a_j}\}$. The \emph{convex cone}
    $\ccone$ generated by $\cP$ consists of polynomials of the form
  \begin{align}\label{eq:weightedSOS}
       q = \sum_{i=1}^M r_i^\dagger r_i +  \sum_{i=1}^N
       \sum_{j=1}^L s_{ij}^\dagger\,p_i\,s_{ij},
  \end{align}
where $p_i \in \cP$, $M$, $N$ and $L$ are finite, and $r_i$, $s_{ij}$
are arbitrary polynomials.
\end{definition}
\noindent In the following, we will call Eq.~\eqref{eq:weightedSOS} a
\emph{weighted sum of squares} (WSOS) representation of $q$.

The purpose of the set $\cP$ is to keep track of the constraints on
the measurement operators. Note that when considering the measurement
operators for non-local games, it is sufficient for us to restrict
ourselves to considering (measurement) operators that are positive
semidefinite. In particular, this means that all operators are
Hermitian. The Positivstellensatz as such does not require us to deal
only with Hermitian variables in the polynomials, but allows us to
use any noncommuting matrix variables. In the following, we will
always take our measurement operators to be of the form $X_{s}^{a} =
(\hat{X}_{s}^{a})^\dagger \hat{X}_{s}^{a}$. Clearly, $X_{s}^a$ is
itself a Hermitian polynomial in the variable $\hat{X}_s^a$. For
clarity of notation, we will omit this explicit expansion in the
future. Note that we will thus not impose the constraint that our
operators are Hermitian, and this
implicit expansion does not increase the size of our SDP.

We can write our constraints in terms of the following sets of
Hermitian polynomials. In the following, we will use the short hand
notation $O_{-j} := X_{s_1}^{a_1}\ldots X_{s_{j-1}}^{a_{j-1}}
X_{s_{j+1}}^{a_{j+1}} \ldots X_{s_N}^{a_N}$ where we leave indices
$s_j$ and $a_j$ implicit, to refer to a product of measurement
operators where we exclude player $j$. First, we want measurements on
different subsystems to commute. In the multi-party case, this gives
us the set of polynomials
$$
\cQ_1 = \{i[X_{s_j}^{a_j},O_{-j}] \mid \mbox{ for all } s_j \in S_j,
a_j \in A_j \mbox{ and all } O_{-j}\}.
$$
Second, we want our operators to form valid measurements.
$$
\cQ_2 = \bigcup_{j,s_j} \{\id - \sum_{a_j} X_{s_j}^{a_j}\}.
$$
Finally, by Neumark's theorem~\cite{A.Peres:Book:1995}, we may take
our measurement operators to be projectors, this gives
$$
\cQ_3 = \bigcup_{j,s_j,a_j} \{(X_{s_j}^{a_j})^2 - X_{s_j}^{a_j}\}.
$$
It's not hard to see that these constraints actually give us orthogonality
of the projectors.
For clarity, however, we may also include the following sets of polynomials
$$
\cQ_4 = \{i[X_{s_j}^{a_j}, X_{s_j}^{a_j'}] \mid
\text{for all $s_j\in S_j$ and all $a_j\neq a_j'$}\},
$$
$$
\cQ_5 = \{X_{s_j}^{a_j}X_{s_j}^{a_j'}+X_{s_j}^{a_j'}X_{s_j}^{a_j}\mid
\text{for all $s_j\in S_j$ and all $a_j\neq a_j'$}\},
$$
which explicitly demand that projectors corresponding to the same
$s_j$ are orthogonal.

Let $\cQ = \cQ_1 \cup \cQ_2 \cup \cQ_3 \cup \cQ_4 \cup \cQ_5$ and let
$\cP = \cQ \cup (-\cQ)$. Note that all polynomials in $\cP$ are
Hermitian. It is clear that if the measurement operator satisfy the
constraints, then the term
$$
\sum_{i,j} s_{ij}^\dagger\, p_j\,s_{ij}
$$
vanishes for arbitrary $p_j\in\cP$ and arbitrary polynomial $s_{ij}$.
We are now ready to state the Positivstellensatz:

\begin{theorem}[Positivstellensatz]\label{theorem:satz}
    Let $G=G(\pi, V)$ be an $N$-prover game and let $\ccone$ be the cone
    generated by the set $\cP$ defined above. Set
  \begin{align}\label{Dfn:qnu}
    q_\nu = \nu \id -\!\!\sum_{s_1,\ldots,s_N}\!\!\! \pi(s_1,\ldots,s_N)
    \sum_{a_1,\ldots,a_N}\!\!\! V(a_1,\ldots,a_N|s_1,\ldots,s_N)
    X_{s_1}^{a_1}\ldots X_{s_N}^{a_N}.
  \end{align}
If $q_\nu > 0$, then $q_\nu \in \ccone$, i.e.,
\begin{align}\label{Eq:qnu:SOS}
    \nu \id -\!\!\sum_{s_1,\ldots,s_N}\!\!\! \pi(s_1,\ldots,s_N)\!\!
    \sum_{a_1,\ldots,a_N}\!\!\!V(a_1,\ldots,a_N|s_1,\ldots,s_N)
    X_{s_1}^{a_1}\ldots X_{s_N}^{a_N}
    =\sum_i r_i^\dagger r_i +  \sum_{i,j} s_{ij}^\dagger\,p_i\,s_{ij},
\end{align}
for some $p_i\in\cP$, and some polynomials $r_i$, $s_{ij}$.
\end{theorem}

\section{Finding upper bounds}\label{findingBounds}

We now show how we can approximate the optimal field-theoretic value
of a non-local game using semidefinite programming. We thereby
construct a converging hierarchy of SDPs, where each level in this
hierarchy gives us a better upper bound on the actual value of the
game. To this end we will use the Positivstellensatz of
Theorem~\ref{theorem:satz} in combination with the beautiful approach
of Parrilo~\cite{parrilo:thesis,parrilo:paper}. For simplicity, we
first describe everything for the two party setting; a generalization
is straightforward.

Recall from Definition~\ref{Dfn:ftv} that if for some real number
$\nu$ we have
\begin{equation}\label{Eq:qnu>0}
    q_\nu = \nu \id - \sum_{a,b,s,t} \pi(s,t) V(a,b|s,t)
    \aoq \boq \geq 0,
\end{equation}
and the operators $\{\aoq\}$ and $\{\boq\}$ form a valid measurement,
then $\nu \geq \omega^f(G)$ gives us an upper bound for the optimum
value of the game. When trying to find the optimal value of the game,
our task is thus to find the smallest $\nu$ for which $q_\nu \geq 0$
for any choice of measurement operators. Clearly, if we could express $q_\nu$ as an SOS for any choice of measurement operators $\{\aoq\}$ and $\{\boq\}$ then $q_\nu \geq 0$ and we would also have $\nu \geq \omega^f(G)$. Luckily, the Positivstellensatz of Theorem~\ref{theorem:satz} gives us almost the converse: \emph{if} $q_\nu > 0$, then $q_\nu$ can be written as a {\em weighted} sum of squares (WSOS). Recall from the previous section, that the purpose of the additional term in the weighted sums of squares representation is to deal with the constraint that we would like the operators $\{\aoq\}$ and $\{\boq\}$ to form a valid quantum measurement. Note that $q_\nu$ reduces to an SOS if we could express $q_\nu$ as a WSOS, i.e.,
\begin{equation}
    q_\nu = \nu \id - \sum_{a,b,s,t} \pi(s,t) V(a,b|s,t)=\sum_{i=1}^M r_i^\dagger r_i +
    \sum_{i=1}^N\sum_{j=1}^L s_{ij}^\dagger\, p_i\, s_{ij},
\end{equation}
for some polynomials $r_i$ and $s_{ij}$ in the variables $\{\aoq\}$
and $\{\boq\}$ in such a way that whenever the variables satisfy the constraints the second term in the above expansion vanishes.

It is not difficult to see that if $q_\nu > 0$, then
there exists \emph{no} strategy that achieves a winning
probability of $\nu$ or higher.
Applied to
our problem, the Positivstellensatz thus tells us that if there
exists \emph{no} strategy that achieves winning probability $\nu$,
then $q_\nu$ \emph{can} be written as a weighted sum of squares.
Intuitively, the WSOS representation of $q_\nu$ thus bears witness to
the fact that the set of measurement operators and states giving a
success probability higher than $\nu$ is empty. The advantage of this
procedure is that semidefinite programming can be used to test
whether polynomials (such as $q_\nu$) admit a representation as WSOS.
In Section~\ref{Sec:Examples}, we will look at some specific examples
of this approach (see also~\cite{parrilo:paper,sostools:manual} for
the analogous treatment for commutative variables).

When trying to find the optimal value of the game, our task is thus
to find the smallest $\nu$ for which $q_\nu$ admits a WSOS
representation. Hence, we want to
\begin{sdp}{minimize}{$\nu$}
&$q_\nu \in \ccone.$
\end{sdp}
Recall that if $q_\nu \in \ccone$, then $q_\nu$ is of the
form
\begin{equation}\label{Eq:qSOS}
    q_\nu =  \sum_{i=1}^M r_i^\dagger r_i +
    \sum_{i=1}^N\sum_{j=1}^L s_{ij}^\dagger\, p_i\, s_{ij},
\end{equation}
for some polynomials $r_i$ and $s_{ij}$ in the variables $\{\aoq\}$
and $\{\boq\}$. A point that is worth noting now is that in the above
optimization,
Eq.~\eqref{Eq:qSOS} is an identity true for all $\{\aoq\}$,
$\{\boq\}$, rather than an equation that is only true when
$\{\aoq\}$, $\{\boq\}$ correspond to projective measurements. In
this, the additional term is rather similar to the Lagrange
multipliers in more conventional constrained optimizations.

\subsection{SDP hierarchy}
The main difficulty now is that we do not know how large the WSOS
representation of $q_\nu$ has to be. That is, we do not know ahead of
time how large we need to choose the degree of the polynomials in the
representation. The techniques discussed above are therefore not
constructive and do not lead to a direct computation of
$\omega^f(G)$. However it is straightforward to find semidefinite
relaxations that provide upper bounds on $\omega^f(G)$. In this we
simply apply the methods of Parrilo~\cite{parrilo:thesis,
parrilo:paper} for the case of polynomials of commutative variables.
The main requirement is to fix an integer $n$ and look for a sum of
squares decomposition for $q_\nu$ that has a total degree of at most
$2n$. Letting $\nu = \omega^f(G) + \eps$, this means that $\eps$ may
not be made arbitrarily small but will always result in an upper
bound for $\omega^f(G)$. This upper bound can be computed as an SDP
using methods analogous to~\cite{parrilo:paper}. Consider the problem given above for $q_\nu$ as in Eq.~\eqref{Eq:qSOS}. Notice that all of the constraint polynomials $p_i$ defined in Section \ref{section:satz} have total degree less than or equal to 2 so we require that each $r_i$ is of total degree $n$ and each $s_i$ is of total degree at most $n-1$. The lowest level of the hierarchy has $n=1$ and corresponds to applying the method of Lagrange multipliers to finding the quantum value of the game. In the following, we use the term \emph{level $n$} to refer to a level of the hierarchy where the total degree of $q_\nu$ is $2n$ when concerned with a bipartite game. For a game $G$, denote the solution to the SDP at level $n$ as $\omega^{\text{sdp}}_{n}(G)$. It should be clear that if $q_\nu$ has a WSOS decomposition of degree $2n$, it must also have a WSOS decomposition with higher degree. As such, the optimum derived from the hierarchy of SDPs must obey the
following inequalities:
\begin{equation}\label{Eq:Omega:Convergence}
    \omega^\text{sdp}_1(G)\geq \omega^\text{sdp}_2(G) \geq \cdots \geq
    \omega^{\text{sdp}}_n(G).
\end{equation}

\begin{theorem}
  \label{theorem:4}
  The solutions to the SDP hierarchy converge to $\omega^{f}(G)$, i.e.,
  $\lim_{n\to\infty}\omega^{\text{sdp}}_n = \omega^f(G)$.
\end{theorem}

\begin{proof}
That $\omega^\text{sdp}_n (G) \geq \omega^f(G)$ follows from our
discussion above. To prove convergence, we use the Positivstellensatz
given by Theorem~\ref{theorem:satz}. Fix $\eps>0$ and let $\nu =
\omega^f(G) + \eps$ with $q_\nu$ defined as in Eq.~\eqref{Dfn:qnu}.
By Theorem~\ref{theorem:satz}, $q_\nu$ has a representation as a
WSOS,
\begin{equation*}
    q_\nu =  \sum_{i=1}^M r_i^\dagger r_i +
    \sum_{i=1}^N\sum_{j=1}^L s_{ij}^\dagger\, p_i\, s_{ij},
\end{equation*}
Let $2D$ be the maximum degree of any of the polynomials $r_i^\dagger
r_i$ and $s_{ij}^\dagger\, p_i\, s_{ij}$ that occurs in the above
expression. Then, if we consider a level $D$ SDP relaxation, we must
necessarily arrive at an optimum such that $\omega_{D}^\text{sdp}(G)
\le \omega^f(G)+\eps$. Likewise, by choosing $\eps$ arbitrarily close
to zero, there is a corresponding SDP with total degree $2n$ whose
optimum $\omega_{n}^\text{sdp}(G)$ is arbitrarily close to
$\omega^f(G)$. In particular, from Eq.~\eqref{Eq:Omega:Convergence},
we can see that as $n\to\infty$, the optimum of the SDP hierarchy
must converge to $\omega^f(G)$.
\end{proof}

In Section~\ref{Sec:I3322}, we provide a simple example of how the
degree of the polynomials can be increased when going from level 1 to
level 2. There are many connections of this semidefinite programming
hierarchy to other methods that can be used to bound the quantum
values of games. In particular, it can be shown that the dual
semidefinite programs to this hierarchy are equivalent to the moment
matrix methods of NPA~\cite{navascues:_bound}, thus showing that the
hierarchy of semidefinite programs discussed in that work converges
to the entangled value of the game.
Our example of the CHSH inequality below demonstrates this connection
explicitly. In relation to this, it is worth noting that the duality
between the two approaches (sum of squares and moment matrix) arises
also in the case of commutative variables where the moment matrix
methods of Laserre~\cite{lasserre01} are dual to the semidefinite
programs discussed by Parrilo~\cite{parrilo:paper}.

\subsection{Examples}\label{Sec:Examples}

\subsubsection{CHSH inequality}\label{example:CHSH}

We will now look at the simplest non-local game that is derived from
the CHSH inequality~\cite{chsh:nonlocal}. In particular, we will
illustrate how the tools that have we developed allow us to
prove that~\cite{tsirel:original}
$$
\mS_{\CHSH} = \ext{A_1 B_1} + \ext{A_1 B_2} + \ext{A_2 B_1} - \ext{A_2 B_2} \leq 2 \sqrt{2},
$$
where $A_1,A_2$ and $B_1,B_2$ are observables with eigenvalues $\pm
1$ corresponding to Alice and Bob's measurement settings
respectively. First of all, note that since we are only interested in
the expectation values of the form $\ext{A_1 B_2}$ we may simplify
our problem: instead of dealing with the probabilities of individual
measurement outcomes, we are only interested in whether said expectation values can be obtained. Here, our constraints become much simpler and we only demand that $A_j^2 = \id$, $B_j^2 = \id$ and $[A_j,B_k] = 0$ for all $j,k \in \{1,2\}$. The Bell operator for the CHSH inequality is given by~\cite{S.L.Braunstein:PRL:1992}
$$
\cB_{\CHSH} = A_1 B_1 + A_1 B_2 + A_2 B_1 - A_2 B_2.
$$
Hence, to find the optimum value our goal is to
\begin{sdp}{minimize}{$\nu$}
&$q_\nu = \nu\id - \cB_{\CHSH} \in \ccone$.
\end{sdp}
The constraint in the above optimization thereby amounts to
determining whether $q_\nu$ as written above can be cast in the form
of a WSOS, which reduces to an SOS for measurement operators
satisfying the constraints. The numerical package
SOSTOOLS~\cite{sostools} gives a frontend to other SDP solvers and
explains how to apply these techniques in the case of commutative
variables. Similar ideas can be applied here. However, for our simple
example, it is not hard to see how this problem can be recast in a
language that may be more familiar.

Since $\cB_{\CHSH}$ is a noncommutative
polynomial of degree 2, the lowest level relaxation consists of
looking for a WSOS decomposition for $q_\nu$ that is of degree 2. To
this end, we shall consider a vector of monomial of degree 1, namely,
$z = (A_1,A_2,B_1,B_2)^\dagger$. Our goal is to find a
$4\times 4$ matrix $\Gamma$ such that $q_\nu = z^\dagger \Gamma z$
whenever the constraints are satisfied. I.e. we have
$[A_j,B_k] = 0$ for all $j,k \in \{1,2\}$, and
polynomials
\begin{equation}\label{Eq:p_j:CHSH}
    p^{(A)}_j:=\id- (A_j)^2,\quad
    p^{(B)}_j:=\id- (B_j)^2,\quad j=1,2,
\end{equation}
and their negations vanish.
Evidently, since we want $q_\nu$ to be a
Hermitian polynomial, and we want our commutation constraints to hold, we may without loss of generality take $\Gamma$
to be real and symmetric.
Note that this already takes care of the commutation constraints.
Moreover, since all remaining constraints are
quadratic, when looking for a WSOS decomposition for $q_\nu$, it
suffices to consider $s_{ij}$ in Eq.~\eqref{Eq:qSOS} as multiples of
$\id$. Let $\gamma_{ij} = [\Gamma]_{i,j}$, then
a small calculation shows that this amounts to requiring
\begin{eqnarray}
\nu &=& \gamma_{11} + \gamma_{22} + \gamma_{33} + \gamma_{44}\nonumber\\
0 &=& \gamma_{12} = \gamma_{21} = \gamma_{34} = \gamma_{43}\nonumber\\
-1 &=& 2\gamma_{13} = 2\gamma_{14} =2\gamma_{23}\nonumber\\
1 &=&  2\gamma_{24},
\label{Eq:constraints:gamma}
\end{eqnarray}
so that
\begin{equation}\label{Eq:qnu:CHSH:explicit}
    q_\nu=\nu\,\id-\cB_{\CHSH}=z^\dagger\Gamma z+\sum_{j=1}^2
    \gamma_{jj}\, p_j^{(A)} +\sum_{j=3}^4 \gamma_{jj}\,
    p_{j-2}^{(B)}.
\end{equation}
Using the constraints given in Eq.~\eqref{Eq:constraints:gamma}, we
see that $\Gamma$ should be of the form
\begin{equation}\label{Eq:Gamma:CHSH}
    \Gamma = \frac{1}{2}\left(\begin{array}{cccc}
    2\gamma_{11} &0 &-1 &-1\\
    0 & 2\gamma_{22} & -1 & 1\\
    -1 & -1 & 2\gamma_{33} & 0\\
    -1 & 1 & 0 & 2\gamma_{44}
    \end{array}\right).
\end{equation}
Effectively, $\Gamma$ is the matrix obtained by expressing
$\nu\,\id-\cB_{\CHSH}-\sum_{j=1}^2\gamma_{jj}\, p_j^{(A)}
-\sum_{j=3}^4 \gamma_{jj}\, p_{j-2}^{(B)}$ in the form of
$z^\dagger\Gamma z$. Now, if we can find a $\Gamma \geq 0$ that is of
this form, then whenever the polynomials given in
Eq.~\eqref{Eq:p_j:CHSH} vanish, $q_\nu = z^\dagger \Gamma z$ is an
SOS. To see this, note that in this case, we may write $\Gamma =
U^\dagger D U$, where $U$ is unitary and $D = \text{diag}(d_i)$ only
consists of nonnegative diagonal entries. Then we can write $q_\nu$
as $\sum_i d_i (Uz)_i^\dagger (Uz)_i$ which is clearly an SOS.
Conversely, note that if $q_\nu$ is an SOS, we can find such a matrix
$\Gamma$. Hence, we can rephrase our optimization problem as the SDP
\begin{sdp}{minimize}{$\Tr(\Gamma)$}
&$\Gamma \geq 0$.
\end{sdp}
This is, in fact, exactly the dual of the SDP corresponding to the first level of the SDP hierarchy given by NPA~\cite{navascues:_bound}, and the dual of the SDP for the special case of XOR games~\cite{wehner05d}. Solving this SDP, one
obtains
\begin{equation}\label{CHSHgamma}
      \Gamma = \frac{1}{2}\begin{pmatrix}
     \sqrt{2} & 0 & -1 & -1 \\
     0 & \sqrt{2} & -1 & 1 \\
     -1 & -1 & \sqrt{2} & 0 \\
     -1 & 1 & 0 & \sqrt{2}
    \end{pmatrix},
\end{equation}
which gives $2\sqrt{2}$ as an optimum. From here and
Eq.~\eqref{Eq:qnu:CHSH:explicit}, it is possible to write down a WSOS
decomposition for $\nu=2\sqrt{2}$ as
\begin{align}\label{eqNext:2}
    q_{2\sqrt{2}} = 2\sqrt{2}\,\id - \cB_{\CHSH} =
    \frac1{2\sqrt{2}} (h_1^\dagger h_1 + h_2^\dagger h_2)
    +\frac{1}{\sqrt{2}}\sum_{j=1}^2 p_j^{(A)}
    +\frac{1}{\sqrt{2}}\sum_{j=3}^4  p_{j-2}^{(B)},
\end{align}
with $h_1 = A_1 + A_2 - \sqrt{2}\,B_1$ and $h_2 = A_1 - A_2
-{\sqrt{2}}\,B_2$. This immediately implies that whenever the
constraints are satisfied, $q_{2\sqrt{2}} \geq 0$ and hence
$\cB_{\CHSH} \leq 2\sqrt{2}\,\id$. It is well known that for
the CHSH inequality, this bound can be
achieved~\cite{tsirel:original}.

\subsubsection{The $I_{3322}$ inequality}\label{Sec:I3322}

We now consider another example of a two-player game, where the first level of the
hierarchy does not give a tight bound.
The $I_{3322}$ inequality~\cite{D.Collins:JPA:2004} is a Bell
inequality phrased in terms of probabilities (not expectation values)
whereby Alice and Bob can each perform
one of three possible two-outcome measurements. Without loss of generality, the Bell
operator in this case can be written as:
\begin{align*}
    \cB_{3322}=A^a_1(B^b_1+B^b_2+B^b_3)+A^a_2(B^b_1+B^b_2-B^b_3)
    +A^a_3(B^b_1-B^b_2)-A^a_1-2B^b_1-B^b_2,
\end{align*}
where $A^a_i$ and $B^b_j$ ($i,j=1,2,3$) are projectors corresponding
to, respectively, outcome $a$ of Alice's $i$-th measurement and
outcome $b$ of Bob's $j$-th measurement for some fixed $a$ and $b$.
To the best our knowledge, the maximum entangled value for
$\cB_{3322}$, i.e., $\omega^*(I_{3322})$, is not known. The best
known lower bound on $\omega^*(I_{3322})$ is
0.25~\cite{D.Collins:JPA:2004}; some upper bounds
(0.375~\cite{Y.C.Liang:PRA:2007}, 0.3660~\cite{D.Avis:JPA:2006}) are
also known in the literature. Here, we will make use of the tools
that we have developed  to obtain a hierarchy of upper bounds on this
maximum. In analogous with the CHSH scenario, this corresponds to
solving the following SDP for some fixed degree of $q_\nu$:
\begin{align}
    &\text{minimize\ \ \ } \nu,\nonumber\\
    &\text{subject to\ \ } q_\nu=\nu\,\id-\cB_{3322}\in\ccone.
    \label{Eq:SDP:abstract}
\end{align}

In particular, since $\cB_{3322}$ is a noncommutative polynomial of
degree 2, the lowest level SDP relaxation would correspond to
choosing a vector of monomials with degree at most one, i.e.,
\begin{equation}\label{Eq:z:1st:3322}
    z^\dagger=(\id, A^a_1, A^a_2, A^a_3, B^b_1, B^b_2, B^b_3).
\end{equation}
We can now proceed analogously to the CHSH case, where we will look
for a particular matrix $\Gamma$ restricted by our constraints,
namely, $(A^a_j)^2 = A^a_j$ and $(B^b_j)^2 = B^b_j$ for all $j \in
\{1,2,3\}$, where again for the purpose of implementation, we will
implicitly enforce the commutativity conditions $[A_i,B_j]=0$ for all
$i,j\in\{1,2,3\}$. Solving the corresponding SDP (Appendix~\ref{Sec:SDP:Lowest}), one obtains
$\omega_1^\text{sdp}(I_{3322})=3/8$, and the matrix
\begin{gather*}
    \Gamma=\frac{1}{2}\left(\begin{array}{rrrrrrr}
    \frac{3}{4} & 0 & -1 & -\frac{1}{2} & 1 &
    0 & -\frac{1}{2}\\
    0 & 2 & 0 & 0 & -1 & -1 & -1\\
    -1& 0 & 2 & 0 & -1& -1& 1\\
    -\frac{1}{2} & 0 & 0 & 1& -1& 1& 0\\
    1& -1& -1& -1& 2 & 0 & 0\\
    0 & -1& -1& 1& 0 & 2 & 0\\
    -\frac{1}{2} & -1& 1& 0 & 0 & 0 & 1\\
    \end{array}\right),
\end{gather*}
which provides a WSOS decomposition for $\nu=3/8$, i.e.,
\begin{equation}\label{Eq:3322:SOS}
    q_{3/8} = \frac{3}{8}\,\id-\mB_{3322}=z^\dagger\Gamma z +\sum_i s_i^\dagger\,p_i s_i,
\end{equation}
where
\begin{gather}\label{Eq:p_i:3322}
    p_i=\left\{ \begin{array}{c@{\quad:\quad}l}
        A^a_i-\left(A^a_i\right)^2 & i=1,2,3,\\
        B^b_{i-3}-\left(B^b_{i-3}\right)^2 & i=4,5,6.\\
        \end{array} \right. ,\quad
    s_i=\left\{ \begin{array}{c@{\quad:\quad}l}
        1  & i=1,2,4,5,\\
        \frac{1}{\sqrt{2}} & i=3,6.\\
        \end{array} \right. .
\end{gather}

Given that $\omega_1^\text{sdp}(I_{3322})$ is far from the best known lower
bound on $\omega^*(I_{3322})$, it seems natural to also look at
higher level relaxations for $I_{3322}$. For the next level, we will
look for $q_\nu$ that is of degree at most 4. Note that we have many
options to extend the hierarchy. The easiest way to proceed is to
extend the vector $z$ by some monomials of degree two in the
measurement operators. For this we do not even have to consider using
all possible degree 2 monomials, but could consider only a subset of
them such as that given by the following vector
\begin{equation}\label{Eq:z:1st+AB}
    z^\dagger=(\id, A^a_1, A^a_2, A^a_3, B^b_1, B^b_2, B^b_3, A^a_1
    B^b_1, A^a_1B^b_2,A^a_1B^b_3,A^a_2B^b_1, A^a_2B^b_2,A^a_2B^b_3,
    A^a_3B^b_1, A^a_3B^b_2, A^a_3B^b_3).
\end{equation}
In particular, solving the corresponding SDP
(Appendix~\ref{Sec:SDP:Higher}) with $z$ given by
Eq.~\eqref{Eq:z:1st+AB} gives an optimum that is approximately
$0.251~470~90$ which is significantly less than $3/8=0.375$.
Clearly, we could increase the size of $z$ further by including all
relevant monomials of degree 2 or less.
\begin{equation*}
    z^\dagger=(\id,A^a_1,A^a_2,\ldots,B^b_2,B^b_3,
    A^a_1A^a_2,A^a_1A^a_3, \ldots, A^a_3A^a_2, B^b_1B^b_2, B^b_1B^b_3,\ldots, B^b_3B^b_2,
    A^a_1B^b_1,A^a_1B^b_2,\ldots,A^a_3B^b_3).
\end{equation*}
Proceeding as before, we end up with the optimum of the second order
relaxation $\omega_2^\text{sdp}(I_{3322})\approx 0.250~939~72$. In the next
level, we would then include all monomials of degree 3 and less, and
this gives $\omega_3^\text{sdp}(I_{3322})\approx 0.250~875~56$.

\subsubsection{Yao's inequality}\label{yao}

Finally, we examine a well-known tripartite Bell
inequality~\cite{yao:inequality} among three provers: Alice, Bob and
Charlie. Each prover performs three possible measurements, each of
which has two possible outcomes. Similarly to the CHSH inequality
above, we may thus express each measurement as an observable with
eigenvalues $\pm 1$. For simplicity, let $A_1,A_2,A_3$, $B_1,B_2,B_3$ and
$C_1,C_2,C_3$ correspond to the observables of Alice, Bob and Charlie
respectively. Yao's inequality states that for any shared state
$\rho$ we have
\begin{equation}\label{Eq:YaoIneq}
    \mS_{\Yao} = \ext{A_1 B_2 C_3} + \ext{A_2 B_3 C_1} +
    \ext{A_3 B_1 C_2} - \ext{A_1 B_3 C_2} - \ext{A_2 B_1 C_3} - \ext{A_3
    B_2 C_1} \leq 3 \sqrt{3}.
\end{equation}

We now provide a very simple proof of this inequality based on our
framework. First of all, note that since we are only interested in
the expectation values of the form $\ext{A_1 B_2 C_3}$ we may again
restrict ourselves to dealing only with expectation values in analogy
with the CHSH example presented above. Our constraints are also
analogous to the CHSH case. Among them, we have the following
constraint polynomials
\begin{equation}\label{Eq:Constraint:Yao}
    p^{(A)}_j:=\id- (A_j)^2,\quad
    p^{(B)}_j:=\id- (B_j)^2,\quad
    p^{(C)}_j:=\id- (C_j)^2,
\end{equation}
for $j=1,2,3$. Next, note that the Bell operator for Yao's inequality
can be written as [c.f. Eq.~\eqref{Eq:YaoIneq}]
\begin{equation}\label{Eq:BellOperator:Yao}
    \mB_{\Yao} = A_1 B_2 C_3 + A_2 B_3 C_1 + A_3 B_1 C_2 - A_1 B_3 C_2 -
    A_2 B_1 C_3 - A_3 B_2 C_1,
\end{equation}
which is a noncommutative polynomial of degree 3.

For our task at hand, we will consider the following SDP relaxation
\begin{align*}
    &\text{minimize\ \ \ } \nu,\\
    &\text{subject to\ \ } q_\nu=\nu\,\id-\cB_{\Yao}\in\ccone,
\end{align*}
with $q_\nu$ being a polynomial of degree at most 6. As usual, we
will implicitly enforce the commutativity constraints, i.e.,
$[A_i,B_j]=0$, $[A_i,C_k]=0$, and $[B_j,C_k]=0$ for all
$i,j,k\in\{1,2,3\}$. With this assumption, it turns out that it
suffices to consider the following 25-element vector
\begin{equation}\label{Eq:z:NewBasis}
    z=
    \left(\begin{array}{c}
     \id     \\
     A_1B_2C_3 \\
     A_2B_3C_1 \\
     A_3B_1C_2 \\
     A_1B_3C_2 \\
     A_2B_1C_3 \\
     A_3B_2C_1 \\
     \end{array}\right)\oplus
     \left(\begin{array}{c}
     A_1B_1C_2 \\
     A_1B_2C_1 \\
     A_2B_1C_1 \\
     A_1B_2C_2 \\
     A_2B_1C_2 \\
     A_2B_2C_1 \\
     \end{array}\right)\oplus
     \left(\begin{array}{c}
     A_1B_1C_3 \\
     A_1B_3C_1 \\
     A_3B_1C_1 \\
     A_1B_3C_3 \\
     A_3B_1C_3 \\
     A_3B_3C_1 \\
     \end{array}\right)\oplus
     \left(\begin{array}{c}
     A_2B_2C_3 \\
     A_2B_3C_2 \\
     A_3B_2C_2 \\
     A_2B_3C_3 \\
     A_3B_2C_3 \\
     A_3B_3C_2 \\
     \end{array}\right).
\end{equation}
In this case, since the constraint polynomials given in
Eq.~\eqref{Eq:Constraint:Yao} are quadratic, when looking for a WSOS
decomposition for $q_\nu$, we will need to consider $s_{ij}$ in
Eq.~\eqref{Eq:qSOS} as an arbitrary polynomial of $A_i$, $B_j$ and
$C_k$ with degree at most 2. Proceeding in a way analogous to the 2nd
level relaxation for $I_{3322}$ inequality, one obtains the
$25\times25$ positive semidefinite matrix
$    \Gamma=\Gamma_{7\times7}\bigoplus_{i=1}^6\Gamma_{3\times3}$,
where
\begin{gather}\label{Eq:Gamma:Blks}
    \Gamma_{7\times7}:=\hf\left(\begin{array}{ccccccc}
    3\sqrt{3} & -1 & -1 & -1 & 1 &
    1 & 1\\
    -1 & \frac{1}{\sqrt{3}} & 0 & 0 & -\frac{1}{3\sqrt{3}} & -\frac{1}{3\sqrt{3}} & -\frac{1}{3\sqrt{3}}\\
    -1 & 0 & \frac{1}{\sqrt{3}} & 0 & -\frac{1}{3\sqrt{3}} & -\frac{1}{3\sqrt{3}} & -\frac{1}{3\sqrt{3}}\\
    -1 & 0 & 0 & \frac{1}{\sqrt{3}} & -\frac{1}{3\sqrt{3}} & -\frac{1}{3\sqrt{3}} & -\frac{1}{3\sqrt{3}}\\
    1 & -\frac{1}{3\sqrt{3}} & -\frac{1}{3\sqrt{3}} & -\frac{1}{3\sqrt{3}} & \frac{1}{\sqrt{3}} & 0 & 0\\
    1 & -\frac{1}{3\sqrt{3}} & -\frac{1}{3\sqrt{3}} & -\frac{1}{3\sqrt{3}} & 0 & \frac{1}{\sqrt{3}} & 0\\
    1 & -\frac{1}{3\sqrt{3}} & -\frac{1}{3\sqrt{3}} & -\frac{1}{3\sqrt{3}} & 0 & 0 & \frac{1}{\sqrt{3}}\\
    \end{array}\right),\quad
    \Gamma_{3\times3}:=\frac{1}{12\sqrt{3}}\left(\begin{array}{ccc}
    1 & 1 & 1\\
    1 & 1 & 1\\
    1 & 1 & 1
    \end{array}\right).
\end{gather}
From some simple calculations, it then follows that
whenever the constraints $A_i^2=B_j^2=C_k^2=\id$ are satisfied, we
have
\begin{equation}\label{Eq:Yao:SOSwoConstraints}
    3\sqrt{3}\,\id - \mB_{\Yao}=z^\dagger\Gamma z=\frac{1}{6\sqrt{3}} \left(h_0^\dagger h_0
    +\sum_{j=1}^2 h_{+,j}^\dagger h_{+,j}
    +\sum_{j=1}^2 h_{-,j}^\dagger h_{-,j}
    +\hf\sum_{j,k=1,2,3} h_{j,k}^\dagger h_{j,k}\right),
\end{equation}
where
\begin{align*}
    h_0&=3\sqrt{3}\,\id - \mB_{\Yao},\\
    h_{+,j}&=A_1B_2C_3 +\exp{i(2\pi j/3)}A_2B_3C_1 +\exp{i(4\pi j/3)}A_3B_1C_2,\\
    h_{-,j}&=A_1B_3C_2 +\exp{i(2\pi j/3)}A_2B_1C_3 +\exp{i(4\pi j/3)}A_3B_2C_1,\\
    h_{j,k}&=A_jB_jC_k + A_jB_kC_j + A_kB_jC_j.
\end{align*}
This makes it explicit that whenever the constraints are satisfied,
$3\sqrt{3}\,\id -\mB_{3322}\ge0$ and therefore
$\mS_{3322}\le3\sqrt{3}$. As a last remark, we note that the
constraint term $\sum_{i,j} s_{ij}^\dagger\, p_i\,s_{ij}$ could have
been included explicitly in the WSOS decomposition for
$q_{3\sqrt{3}}=3\sqrt{3}\,\id-\mB_{\Yao}$ and we refer the reader to
Appendix~\ref{App:Yao} for details.

\section{Acknowledgements}
We thank Stefano Pironio and Tsuyoshi Ito for interesting discussions and for
sharing drafts of their papers~\cite{navascues08:long,Ito08} with us.
SW is supported by the National Science Foundation under contract
number PHY-0456720.

\appendix

\section{Tool 1: Tensor product structure from commutation relations}
\label{Sec:Tensor}

We now provide a simple proof of Lemma~\ref{tensorProduct} from a
computer science perspective that is suitable to the task at hand.
For simplicity, we address the case of two-prover systems in detail,
and merely sketch the extension to the multiple provers at the end.
For ease of reference, we shall now reproduce the Lemma for the two-prover setting:

\begin{lemma}
Let $\hil$ be a finite-dimensional Hilbert space, and let $\{\aoq \in
\bop(\hil)\mid s \in S\}$ and $\{\boq \in \bop(\hil)\mid s \in T\}$.
Then the following two statements are equivalent:
\begin{enumerate}
\item For all $s \in S$, $t \in T$, $a \in A$ and $b \in B$ it holds that
$[\aoq,\boq] = 0$.
\item There exist Hilbert spaces $\hil_A, \hil_B$ such that $\hil
    = \hil_A \otimes \hil_B$ and for all $s \in S$, $a \in A$ we
    have $\aoq \in \bop(\hil_A)$ and for all $t \in T$, $b \in B$
    we have $\boq \in \bop(\hil_B)$.
\end{enumerate}
\end{lemma}

For our argument we will not consider individual operators, but
instead look at the $C^*$-algebra of operators which is well
understood in finite dimensions~\cite{takesaki,robinson}. The
$C^*$-algebra of operators $\setA = \{A_1,\ldots,A_n\}$ consists of
all complex polynomials in such operators and their conjugate
transpose: if $A$ is an element of the algebra, then so is
$A^\dagger$. For example, the set of all bounded operators
$\bop(\hil)$ on a Hilbert space $\hil$ is a $C^*$-algebra. For
convenience, we will also write $\algA = \langle \setA \rangle$ for
such an algebra $\algA$ generated by operators from the set $\setA$.
We will need the following notions: The \emph{center} $\algZ$ of an
algebra $\algA$ is the set of all elements in $\algA$ that commute
with all elements of $\algA$, i.e., $\algZ = \{Z | Z \in \algA,
\forall A \in \algA: [Z,A] = 0\}$. If $\algA \subseteq \bop(\hil)$
for some Hilbert space $\hil$, then the \emph{commutant} of $\algA$
in $\bop(\hil)$ is given by $\comm(\algA) = \{X| X \in \bop(\hil),
\forall A \in \algA: [X,A] = 0\}$. Furthermore, an algebra $\algA$ is
called \emph{simple}, if its only ideals\footnote{An ideal $\algI$ of
$\algA$ is a subalgebra $\algI \subseteq \algA$ such that for all $I
\in \algI$ and $A \in \algA$, we have $IA \in \algI$ and $AI \in
\algI$.} are $\{0\}$ and $\algA$ itself. It is easy to see that if
$\algA$ only has a trivial center, i.e., $\algZ = \{c\,\id| c \in
\Complex\}$, then $\algA$ is simple~\cite{takesaki}. Finally, $\algA$
is called \emph{semisimple} if it can be decomposed into a direct sum
of simple algebras.

\subsection{Optimizing non-local games}
Before we show how to prove Lemma~\ref{tensorProduct}, we first
demonstrate that when considering non-local games we can greatly
simplify our problem and restrict ourselves to $C^*$-algebras that
are simple. It is well known that we can decompose any finite
dimensional algebra into the sum of simple algebras.

\begin{lemma}[\cite{takesaki}]\label{finiteMeansSemisimple}
Let $\algA$ be a finite-dimensional $C^*$-algebra. Then there exists
a decomposition
$$
\algA = \bigoplus_j \algA_j,
$$
such that $\algA_j$ is simple.
\end{lemma}

We furthermore note that for any simple algebra, the following holds:
\begin{lemma}[\cite{takesaki}]\label{isomorphBop}
Let $\hil$ be a Hilbert space, and let $\algA \subseteq \bop(\hil)$
be simple, then there exists a bipartite partitioning of the Hilbert
space $\hil$ such that $\hil = \hil_1 \otimes \hil_2$ and $\algA
\isomorph \bop(\hil_1) \otimes \id_2$.
\end{lemma}

We now show that without loss of generality, we may assume that the
algebras generated by Alice and Bob's measurement operators are in
fact simple.

\begin{lemma}\label{simpleIsEnough}
Let $\hil = \hil_A \otimes \hil_B$ and let $\setA = \{\aoq \in
\bop(\hil_A)\}$ and $\setB = \{\boq \in \bop(\hil_B)\}$ be the set of
Alice and Bob's measurement operators respectively. Let $\rho \in
\bop(\hil)$ be the state shared by Alice and Bob. Suppose that for
such operators we have
$$
q = \sum_{s \in S,t \in T} \pi(s,t) \sum_{a \in A,b \in B} V(a,b|s,t)\Tr\left((\aoq \otimes \boq) \rho\right).
$$
Then there exist measurement operators
$\tilde{\setA} = \{\tildeaoq\}$ and
$\tilde{\setB} = \{\tildeboq\}$
and a state $\tilde{\rho}$ such
\begin{equation*}
    q \leq \sum_{s \in S,t \in T} \pi(s,t) \sum_{a \in A,b \in B}
    V(a,b|s,t)\Tr\left((\tildeaoq \otimes \tildeboq) \tilde{\rho}\right).
\end{equation*}
and the $C^*$-algebra generated by $\tilde{\setA}$ and $\tilde{\setB}$ is simple.
\end{lemma}
\begin{proof}
Let $\algA = \langle \setA \rangle$ and $\algB = \langle \setB
\rangle$. If $\algA$ and $\algB$ are simple, we are done. If not, we
know from Lemma~\ref{finiteMeansSemisimple} and
Lemma~\ref{isomorphBop} that there exists a decomposition $\hil_A
\otimes \hil_B = \bigoplus_{jk} \hil_A^j \otimes \hil_B^k$. Consider
$\Tr((M_A \otimes M_B)\rho)$, where $M_A \otimes M_B \in \algA
\otimes \algB$. It follows from the above that $M_A \otimes M_B =
\bigoplus_{jk} (\Pi_A^j \otimes \Pi_B^k) M_A \otimes M_B (\Pi_A^j
\otimes \Pi_B^k)$, where $\Pi_A^j$ and $\Pi_B^k$ are projectors onto
$\hil_A^j$ and $\hil_B^k$ respectively. Let $\hat{\rho} =
\bigoplus_{jk} (\Pi_A^j \otimes \Pi_B^k) \rho (\Pi_A^j \otimes
\Pi_B^k)$. Clearly,
\begin{equation*}
    \Tr((M_A \otimes M_B)\hat{\rho}) = \Tr\left(\bigoplus_{jk}(\Pi_A^j \otimes \Pi_B^k)
    M_A \otimes M_B (\Pi_A^j \otimes \Pi_B^k) \rho\right) =
    \Tr((M_A \otimes M_B)\rho).
\end{equation*}
The statement now follows immediately by convexity: Alice and Bob can
now measure $\rho$ $\{\Pi_A^j \otimes \Pi_B^k\}$ and recording the
classical outcomes $j, k$. The new measurements will then be
$\tilde{A}_{s,j}^a = \Pi_A^j \aoq \Pi_A^j$ and $\tilde{B}_{t,k}^b =
\Pi_B^k \boq \Pi_B^k$ on state $\tilde{\rho}_{jk} = (\Pi_A^j \otimes
\Pi_B^k)\rho(\Pi_A^j \otimes \Pi_B^k)/\Tr((\Pi_A^j \otimes
\Pi_B^k)\rho)$. By construction, $\tilde{\algA}_j =
\{\tilde{A}_{s,j}^a\}$ and $\tilde{\algB}_k = \{\tilde{B}_{t,k}^b\}$
are simple.

Let $q_{jk}$ denote the probability that we obtain outcomes $j, k$, and let
\begin{equation*}
    r_{jk} = \sum_{s \in S,t \in T} \pi(s,t) \sum_{a \in A,b \in B}V(a,b|s,t)
    \Tr(\tilde{A}^{a}_{s,j} \otimes \tilde{B}^{b}_{t,j}\tilde{\rho}_{jk}).
\end{equation*}
Then $q = \sum_{jk} q_{jk} r_{jk} \leq \max_{jk} r_{jk}$. Let $u,v$
be such that $r_{u,v} = \max_{jk} r_{jk}$. Hence, we can skip the
initial measurement and instead use measurements $\tildeaoq=
\tilde{A}_{s,u}^a$, $\tildeboq = \tilde{B}_{t,v}^b$ and state
$\tilde{\rho} = \tilde{\rho}_{u,v}$.
\end{proof}

This easy argument also immediately tells us that when $\algA$ and
$\algB$ are abelian, we can find a classical strategy that achieves
$q$: Just perform the measurement as above. If $\algA$ and $\algB$
are abelian, the remaining state will be one-dimensional and hence
classical.

\subsection{Tensor product structure}

We are now ready to prove Lemma~\ref{tensorProduct}. First, we
examine the case where we are given a simple algebra $\algA \in
\bop(\hil)$, for some Hilbert space $\hil$. We will need the
following version of Schur's lemma.
\begin{lemma}\label{schurLemma}
Let $\setZ$ be the center of $\bop(\hil)$. Then $\setZ = \{c\,\id
|\,c \in \Complex\}$.
\end{lemma}
\begin{proof}
Let $C \in \setZ$ and let $d = \dim(\hil)$. Let $\cE = \{E_{ij} | i,j
\in [d]\}$ be a basis for $\bop(\hil)$, where $E_{ij}:=\outp{i}{j}$
is the matrix of all 0's and a 1 at position $(i,j)$. Since $C \in
\setZ$ and $E_{ij} \in \bop(\hil)$ we have for all $i \in [d]$
$$
CE_{ii}  = E_{ii}C.
$$
Note that $C E_{ii}$ (or $E_{ii}C$) is the matrix of all 0's but
the $i$th column (or row) is determined by the elements of $C$. Hence
all off diagonal elements of $C$ must be 0. Now consider
$$
C (E_{ij} + E_{ji}) = (E_{ij} + E_{ji})C.
$$
Note that $C(E_{ij} + E_{ji})$ (or $(E_{ij} + E_{ji})C$) is the
matrix in which the $i$th and $j$th columns (rows) of $C$ have been
swapped and the remaining elements are 0. Hence all diagonal elements
of $C$ must be equal. Thus there exists some $c \in \Complex$ such
that $C = c\,\id$.
\end{proof}

Using this Lemma, we can now show that
\begin{lemma}\label{simpleCommutant}
Let $C \in \bop(\hil_A \otimes \hil_B)$ be such that for all $B \in
\bop(\hil_B)$ we have
$$
[C,(\id_A \otimes B)]=0
$$
Then there exists an $A \in \bop(\hil_A)$ such that $C = A \otimes \id_B$.
\end{lemma}
\begin{proof}
Let $d_A = \dim(\hil_A)$ and $d_B = \dim(\hil_B)$. Note that we can
write any $C\in \bop(\hil_A \otimes \hil_B)$ as
$$
C = \left(\begin{array}{ccc}
C_{11} &\ldots& C_{1d_A}\\
\vdots & &\vdots\\
C_{d_A1} &\ldots& C_{d_Ad_A}\end{array}\right),
$$
for $d_B \times d_B$ matrices $C_{ij}$. We have $C(\id_A \otimes B) =
(\id_A \otimes B)C$ if and only if for all $i,j \in [d_A]$ $C_{ij}B =
BC_{ij}$, i.e., $[C_{ij},B]=0$. Since this must hold for all $B \in
\bop(\hil_B)$, we have by Lemma~\ref{schurLemma} that there exists
some $a_{ij} \in \Complex$ such that $C_{ij} = a_{ij} \id_B$. Hence
$C = A \otimes \id_B$ with $A = [a_{ij}]$.
\end{proof}

For the case that the algebra generated by Alice and Bob's
measurement operators is simple, Lemma~\ref{tensorProduct} now
follows immediately:
\begin{proof}[Proof of Lemma~\ref{tensorProduct} if $\algA$ is simple]
Let $\algA = \langle \{\aoq\}\rangle \subseteq \bop(\hil)$ be the
algebra generated by Alice's measurement operators. If $\algA$ is
simple, it follows from Lemma~\ref{isomorphBop} that $\algA \isomorph
\bop(\hil_A) \otimes \id_B$ for $\hil = \hil_{A} \otimes \hil_{B}$.
It then follows from Lemma~\ref{simpleCommutant} that for all $t \in
T$ and $b \in B$ we must have $\boq \in \bop(\hil_B)$.
\end{proof}
Thus, we obtain a tensor product structure! Recall that
Lemma~\ref{simpleIsEnough} states that for our application this is
all we need.

In general, what happens if $\algA$ is not simple? Whereas our
argument shows that there always exist measurement operations such
that $\algA$ is simple, the solution found via optimization may not
have this property. We now sketch the argument in the case where the
$\algA$ is semisimple, which by Lemma~\ref{finiteMeansSemisimple} we
may always assume in the finite-dimensional case. Fortunately, we can
still assume that our commutation relations leave us with a bipartite
structure. We can essentially infer this from von Neumann's famous
Double Commutant Theorem~\cite{takesaki,robinson}, partially stated
here.
\begin{theorem}
Let $\algA$ be a finite-dimensional $C^*$-algebra. Then there exists
$\hil = \hil_A \otimes \hil_B$ such that
$$
\algA \isomorph \bigoplus_j \bop(\hil_A^j) \otimes \id_B^j
$$
and
\begin{equation}\label{commutantForm}
\comm(\algA) \isomorph \bigoplus_j \id_A^j \otimes \bop(\hil_B^j).
\end{equation}
\end{theorem}
\begin{proof}(Sketch)
We already know from Lemma~\ref{finiteMeansSemisimple} that $\algA$
can be decomposed into a sum of simple algebras. Clearly, the RHS of
Eq.~\eqref{commutantForm} is an element of $\comm(\algA)$. To see
that the LHS is contained in the RHS, consider the projection
$\Pi^j_A$ onto $\hil_A^j$. Note that $\Pi_A^j \in \algA$, and thus
for any $X \in \comm(\algA)$ we have $[X,\Pi^j_A] = 0$. Hence, we can
write $X = \sum_j (\Pi^j_A \otimes \id_B)X(\Pi^j_A \otimes \id_B)$,
and thus we can restrict ourselves to considering each factor
individually. The result then follows immediately from
Lemma~\ref{simpleCommutant}.
\end{proof}

If we have more than two provers, the argument is essentially
analogous, and we merely sketch it in the relevant case when the
algebra generated by the prover's measurements is simple, since
Lemma~\ref{simpleIsEnough} directly extends to more than two provers
as well. Suppose we have $N$ provers $\mP_1,\ldots,\mP_N$ and let
$\hil$ denote their joint Hilbert space. Let $\algA$ be the algebra
generated by all measurement operators of provers
$\mP_1,\ldots.\mP_{N-1}$ respectively. Then it follows from
Lemma~\ref{simpleCommutant} and Lemma~\ref{isomorphBop} that there
exists a bipartite partitioning of $\hil$ such that $\hil =
\hil_{1,\ldots,N-1} \otimes \hil_N$, $\algA \isomorph
\bop(\hil_{1,\ldots,N-1})$ and for all measurement operators $M$ of
prover $\mP_N$ we have that $M \in \bop(\hil_N)$. By applying
Lemma~\ref{simpleCommutant} recursively we obtain that there exists a
way to partition the Hilbert space into subsystems $\hil = \hil_1
\otimes \ldots \otimes \hil_N$ such that the measurement operators of
prover $\mP_j$ act on $\hil_j$ alone.

In quantum mechanics, we will always obtain such a tensor product
structure from commutation relations, even if the Hilbert space is
infinite-dimensional. Here, we start out with a type-I algebra, the
corresponding Hilbert space and operators can then be obtained by the
famous Gelfand-Naimark-Segal (GNS) construction~\cite{takesaki}, an
approach which is rather beautiful in its abstraction. In quantum
statistical mechanics and quantum field theory, we will also
encounter factors of type-II and type-III. As it turns out, the above
argument does not generally hold in this case, however, there are a
number of conditions that can lead to a similar structure. Sadly, we
cannot consider this case here and merely refer to the survey article
by Summers~\cite{summers:qftIndep}. Note that in quantum mechanics
itself, we thus have $\omega^*(G) = \ftv(G)$.

\section{Tool 2: Positivstellensatz}

Here, we will provide the details for the proof of
Theorem~\ref{theorem:satz}. For ease of reference, we first reproduce
the theorem as follows:
\begin{theorem}
    Let $G=G(\pi, V)$ be an $N$-prover game and let $\ccone$ be the cone
    generated by the set $\cP$ defined in Section~\ref{section:satz}. Set
  \begin{align}
    q_\nu = \nu \id -\!\!\sum_{s_1,\ldots,s_N}\!\!\! \pi(s_1,\ldots,s_N)
    \sum_{a_1,\ldots,a_N}\!\!\! V(a_1,\ldots,a_N|s_1,\ldots,s_N)
    X_{s_1}^{a_1}\ldots X_{s_N}^{a_N}.
    \tag{\ref{Dfn:qnu}}
  \end{align}
If $q_\nu > 0$, then $q_\nu \in \ccone$, i.e.,
\begin{align}
    \nu \id -\!\!\sum_{s_1,\ldots,s_N}\!\!\! \pi(s_1,\ldots,s_N)\!\!
    \sum_{a_1,\ldots,a_N}\!\!\!V(a_1,\ldots,a_N|s_1,\ldots,s_N)
    X_{s_1}^{a_1}\ldots X_{s_N}^{a_N}
    =\sum_i r_i^\dagger r_i +  \sum_{i,j} s_{ij}^\dagger\,p_i\,s_{ij},
    \tag{\ref{Eq:qnu:SOS}}
\end{align}
for some $p_i\in\cP$, and some polynomials $r_i$, $s_{ij}$.
\end{theorem}

We now prove the {\em contrapositive} statement of
Theorem~\ref{theorem:satz}. In particular, we show that if $q_\nu$
has \emph{no} representation as a WSOS, then $q_\nu\not>0$ and there
exist operators and a state on some Hilbert space that achieve
winning probability $\nu$. The proof proceeds in two stages. We first
use the Hahn-Banach theorem to show (nonconstructively) the existence
of a linear function that separates $q_\nu$ from the convex cone
$\ccone$ and then use a GNS construction, as described below.
Unfortunately we will not in general end up with operators on a
finite-dimensional Hilbert space.

We start by establishing some simple facts about $\ccone$.
\begin{lemma}
  \label{lemma:5}
    Let $W$ be the product of some number of variables from the set
    $\{X_{s_j}^{a_j}\}$. Then $\id-W^\dagger W \in \ccone$,
    and $\id - W W^\dagger \in \ccone$.
\end{lemma}
\begin{proof}
    The proof is by induction on $n$, the number of variables in the
    product $W$. For $n = 1$, we have to show that for all $j,a_j,s_j$, $\id -
    (X_{s_j}^{a_j})^2 \in \ccone$.
    We do this for $\id - (\xoq)^2$. Writing
  \begin{align}
    \id - (\xoq)^2 =  \sum_{a'_j \neq a_j} (X^{a'_j}_{s_j})^2
    -\sum_{a'_j}\left[ (X^{a'_j}_{s_j})^2 - X^{a'_j}_{s_j}\right] + \Bigl[\id - \sum_{a'_j} X^{a'_j}_{s_j}\Bigr],
  \end{align}
makes it clear that $\id - (\xoq)^2 \in \ccone$. For $n\geq 1$, write
$W = V \xoq $, where we have assumed without loss of generality that
the element $\xoq$ is rightmost in $W$, and where $V$ is the product
of $n-1$ variables. Then
  \begin{align}
    \id - W^\dagger W = \id - \xoq V^\dagger V \xoq  = \id - (\xoq)^2 + \xoq (\id -V^\dagger V) \xoq.
  \end{align}
Now $\id - (\xoq)^2 \in \ccone$ by the result for $n=1$ and $\id -
V^\dagger V \in \ccone$ by the inductive hypothesis. Moreover, for
any polynomial $r\in\ccone$, and any arbitrary polynomial $s$, it is
easy to see that $s^\dagger\,r\,s\in\ccone$. Hence, this implies that
$\id -W^\dagger W \in \ccone$. The argument for $\id - WW^\dagger$ is
analogous.
\end{proof}

\begin{lemma}
  \label{lemma:4}
    Let $p$ be a Hermitian polynomial. Then there exists a real number $t
    \geq 0$ and an $s \in \ccone$ such that $p= s-t\,\id$.
\end{lemma}
\begin{proof}
The polynomial $p$ is a finite sum of terms of the form $p' =  w^* v
W^\dagger V + w v^* V^\dagger W$, where $V,W$ are products of the
variables, $w,v \in \CC$ and $w^*$ is the complex conjugate of $w$
(likewise for $v^*$). If we can show the result for $p'$, then the
result for general polynomials $p$ follows immediately. To this end,
note that we can write
\begin{align}
  p' = (v^* V^\dagger +  w^* W^\dagger)(v V + w W) -|v|^2 V^\dagger V - |w|^2 W^\dagger W
\end{align}
so that
\begin{align}
  (|v|^2 + |w|^2) \id + p' = (v^* V^\dagger +  w^* W^\dagger)(v V + w W)
  +|v|^2(\id - V^\dagger V) + |w|^2(\id - W^\dagger W)
\end{align}
which is in $\ccone$ by Lemma~\ref{lemma:5}. Taking $t = |v|^2 + |w|^2$,
the result follows for $p'$, which in turn implies the result for
general polynomials $p$.
\end{proof}

We now want to show that if $q_\nu$ is a Hermitian polynomial that
does not lie in $\ccone$, then there exists a linear functional that
separates it from $\ccone$. The following Lemma closely
follows~\cite[Proposition 3.3]{helton04:_posit_for_non_commut_polyn}
where the only difference is that we consider polynomials over
\emph{complex} instead of real Hermitian matrices. Fortunately, the
essential ingredient of
the proof, the Hahn-Banach theorem
also holds in this case. We state entire proof for convenience:

\begin{lemma}
  \label{lemma:7}
Let $M$ be the space of Hermitian polynomials over complex matrices.
Let $q$ be a Hermitian polynomial such that $q \not \in \ccone$.
Then there exists a linear functional $\lambda: M \rightarrow \RR$
such that $\lambda(\ccone)\geq 0$, $\lambda(\id) >0$, $\lambda (q) \leq 0$.
\end{lemma}

\begin{proof}
Let $M$ be the space of Hermitian polynomials over complex matrices.
Let $\mu: M \rightarrow \RR$ be a linear functional defined as
$\mu(p) := \inf\{t>0 : p = s-t\id \text{ for some }s\in \ccone\}$.
Note that by Lemma~\ref{lemma:4}, we can express any $p \in M$ in
this form. Clearly, $\mu$ is a seminorm on $M$. Note that for $q
\not\in \ccone$ we have $\mu(q-\id) \geq 1$ by definition. We
consider now a fixed $q\not\in \ccone$  and let $L$ be the span of
$q-\id$ in $M$, i.e., all Hermitian polynomials $t(q-\id)$ with $t
\in \RR$. Define a linear functional $f: L \to \RR$ so that
$f(t(q-\id)) := t$. It is not hard to see that $f\leq \mu$ on $L$.
Now we make the first nonconstructive step. By the Hahn-Banach
theorem \cite[Theorem~3.3]{rudin73:_funct_analy}, $f$ extends to a
linear functional $F: M \to \RR$ such that $F(p) \leq \mu(p)$ for all
$p \in M$.

We now claim that $\lambda = -F$ satisfies the requirements of the
lemma: First of all, note that we have for all $s \in \ccone$ that
$F(s) - F(q) + 1 = F(s - q + (q-\id)) = F(s - \id) \leq \mu(s - \id)
\leq 1$, where the first equality follows from the linearity of $F$,
and the first inequality follows from $F$ being an extension of $f$.
Hence, $F(s) \leq F(q)$. Clearly, we also have that for all $s \in
\ccone$ and $t
> 0$, $ts \in \ccone$ and hence $tF(s) = F(ts) \leq F(q)$. Thus, if
$s \in \ccone$ then $F(s) \leq 0$ and $F(q) \geq 0$. Hence, $\lambda
= - F$ satisfies $\lambda(\ccone)\geq 0$ and $\lambda(q)\leq 0$ as
required.

It remains to show that $\lambda(\id) > 0$. First of all, note that
$\id \in \ccone$. Suppose that on the contrary we have $\lambda(\id)
= 0$. Let $p \in \ccone$ and note that by Lemma~\ref{lemma:4} we may
write $-p = s - t \id$ for some $t > 0$ and $s \in \ccone$. From $t
\id = s + p$, we have $0 = t \lambda(\id) = \lambda(t \id) =
\lambda(s) + \lambda(p) \geq 0$ and hence $\lambda(s) = \lambda(p) =
0$ for all $p \in \ccone$. Now note that since $\id - q \in M$ we may
write $\id - q = r - s$ for some $r,s \in \ccone$. But then $0 =
\lambda(r) - \lambda(s) = \lambda(\id - q) = 1$ which is a
contradiction.
\end{proof}

The remainder of the proof of Theorem~5.2 is now exactly
identical to~\cite{helton04:_posit_for_non_commut_polyn}, which in itself
is analogous to the famous GNS construction~\cite{takesaki,robinson} that allows us to
find a representation in terms of bounded operators on a Hilbert
space. We here provide a slightly annotated version of this approach in the
hope that it will be more accessible to the present audience.

\begin{theorem}[Helton and McCullough]\label{helton}
Let $M$ be the space of Hermitian polynomials. Let $\lambda: M
\rightarrow \RR$ be a linear functional such that
$\lambda(\ccone)\geq 0$ and $\lambda(\id) >0$. Then there exists a
Hilbert space $\hil$, bounded operators $\{\txoq\}$
on $\hil$,  and a state $\gamma \in \hil$ such that for all $p \in
\mP$ we have $p(\{\txoq\})\geq 0$ and for any
Hermitian $q \in M$,
$$
\bra{\gamma}q(\{\txoq\})\ket{\gamma} = \lambda(q).
$$
\end{theorem}
\begin{proof}
First, we construct a Hilbert space $\hil$ from $M$: For $s,t \in M$,
define
$$
\inp{s}{t} = \frac{1}{2} \lambda(s^\dagger t + t^\dagger s).
$$
It is easy to verify that $\inp{s}{t}$ is symmetric, bilinear and is
also positive semidefinite whenever $s=t$, since $s^\dagger s \in
\ccone$ and hence $\inp{s}{s} = \lambda(s^\dagger s) \geq 0$. Note
that $\inp{\cdot}{\cdot}$ is degenerate, but in order to obtain a
Hilbert space, we need to turn $\inp{\cdot}{\cdot}$ into an inner
product. This can be done in the standard way by 'moding out' the
degeneracy: Consider
$$
\cJ = \{s \in M \mid \inp{s}{s} = 0\}.
$$
It is not difficult to verify that $\cJ$ forms a linear subspace of
$M$ and that $\cJ$ is a left ideal of $M$
We now consider the quotient space
$M/\cJ$, which is the vector space created by the equivalence classes
$$
[x] = \{x + j \mid j \in \cJ\}.
$$
Addition and scalar multiplication are defined via the following
operations inherited from $M$: for $x,y \in M$ and $\alpha \in
\Complex$, we have $[x+y] = [x] + [y]$ and $[\alpha x] = \alpha [x]$.
We can now define the inner product
$$
\inp{[x]}{[y]} = \inp{x}{y}.
$$
It is important to note that this inner product does not depend on
our choice of representative from each equivalence class, and we have
eliminated the degeneracy present earlier. The Hilbert space is now
obtained by forming the completion of $M/\cJ$ with respect to this
inner product.

Second, we now need to show that there exists a representation $\pi: M \rightarrow \bop(\hil)$.
We first define the action of $\pi(x)$ with $x \in M$ on the vectors $[y]$ as
$$
X[y] = [xy],
$$
where we use the shorthand $X = \pi(x)$. It is straightforward to
verify that this definition is again independent of our choice of
representative from each equivalence class, and that $\pi$ is a
homomorphism. For simplicity, we only show boundedness for operators
$\{\xoq\} \in M$. To see that $X=\pi(x)$ for $x \in \{\xoq\}$ is
bounded, note that by Lemma~\ref{lemma:5} we have $\id - x^\dagger x
\in \ccone$ and that
$$
\inp{X[s]}{X[s]} = \inp{xs}{xs} = \inp{s}{s} - \lambda(s^\dagger(\id - x^\dagger x)s),
$$
where $\lambda(s^\dagger (\id - x^\dagger x)s) \geq 0$, since
$s^\dagger (\id - x^\dagger x)s \in \ccone$. From Lemma~\ref{lemma:5}
we also have that $\id - x x^\dagger \in \ccone$, and hence the
argument for $\pi(x^\dagger)$ is analogous and we may write
$X^\dagger = \pi(x^\dagger)$ without ambiguity. Hence we can find
claimed operators $\{\txoq\} \in \bop(\hil)$.

Third, we need to define the vector $\gamma$. Since $\id \in M$ we choose
$\gamma = [\id]$. Hence, $\inp{\gamma}{\gamma} = \lambda(\id) > 0$, and thus
$\gamma$ is non-zero. Let $q \in M$ and write $q(X)$ for the polynomial
where variables $x_j$ have been substituted by their representations
$X_j$. Note that
$$
\inp{q(X)\gamma}{\gamma} = \lambda(q),
$$
where we have used the fact that $q$ is a Hermitian polynomial.
By a similar argument, it follows that for $p \in \ccone$ and $r \in M$ we have
$$
\inp{p(X)[r]}{[r]} = \lambda(r^\dagger p r) \geq 0,
$$
since $r^\dagger p r \in \ccone$ and hence $p(X) \geq 0$ as promised.
\end{proof}

We can now complete the proof of Theorem~\ref{theorem:satz}.
\begin{proof}[Proof of Theorem~\ref{theorem:satz}]
Recall that our goal was to prove the contrapositive: If $q_\nu\notin
\ccone$ then $q_\nu\not > 0$. From Lemma~\ref{lemma:7} we have that
if $q_\nu \notin \ccone$, then there exists a linear functional
$\lambda$ that separates $q_\nu$ from $\ccone$. Lemma~\ref{helton}
gives us that there exists a Hilbert space $\hil$, measurement
operators $\{\tilde\xoq\}$, and a vector $\gamma$
such that
$$
\lambda(q_\nu) = \bra{\gamma}q_\nu(\{\tilde\xoq\})\ket{\gamma} \leq 0.
$$
and the operators satisfy all the constraints, i.e., for all $p \in
\mP$ we have $p(\{\tilde\xoq\}) \geq 0$. (Note that
we only have equality constraints, which we implemented by including
both $p$ and $-p$ in $\cP$.) Since $\gamma$ is not zero, we have
$q_\nu \not > 0$ which completes the proof.
\end{proof}

Unfortunately, Theorem~\ref{helton} does not tell us whether the
underlying Hilbert space $\hil$ is finite-dimensional, or whether the
algebra generated by the operators $\{\xoq\}$ is type-I at
all. Hence, we cannot ensure without further proof that the fact that
our measurement operators do satisfy the commutation constraints
necessarily leads to them having tensor product form. Thus, we do not
know whether there exist games $G$ for which $\omega^*(G) < \ftv(G)$:
for such games we may would have to get a type-II or type-III
algebra.

\section{Tool 3: Semidefinite Programming}

A semidefinite program (SDP) is an optimization over Hermitian
matrices~\cite{L.Vandenberghe:SR:1996}. The objective function
depends linearly on the matrix variable and the optimization is
carried out subjected to the matrix variable being positive
semidefinite and satisfies various affine constraints. Any
semidefinite program may be written in the following {\it standard
form}~\cite{Boyd:04}:
\begin{subequations}\label{Eq:SDP:Matrix}
\begin{align}
    &\text{maximize\ \ } -\Tr\left[ G_0 Z\right], \label{Eq:SDP:Matrix:Obj}\\
    &\text{subject to\ } \quad\,\Tr\left[ G_k Z\right] =b_k \quad \forall~k, \label{Eq:SDP:Matrix:Eq}\\
    &\qquad\qquad\qquad\quad Z\ge 0, \label{Eq:SDP:Matrix:Ineq}
\end{align}
\end{subequations}
where $G_0$ and all the $G_k$'s are Hermitian matrices, and the $b_k$
are real numbers that together specify the optimization; $Z$ is the
Hermitian matrix variable to be optimized.

An SDP also arises in the {\em inequality form}, which seeks  to
minimize a linear function of the optimization variables
$(x_1,x_2,\ldots,x_n)\in\mathbb{R}^n$, subjected to a linear matrix
inequality:
\begin{subequations}\label{Eq:SDP:Vector}
\begin{gather}
\text{minimize \ \ \ \ }\quad b_k'\,x_k\qquad\label{Eq:SDP:Vector:Obj}\\
\text{subject to \ \ }F_0+\sum_k x_k F_k \geq
0.\label{Eq:SDP:Vector:Ineq}
\end{gather}\end{subequations}
As in the standard form, $F_0$ and all the $F_k$'s are Hermitian
matrices, while $(b_1', b_2',\ldots,b_n')$ is a real vector of length
$n$.

\section{Some Other Miscellaneous Details}

\subsection{Implementing Lowest Level SDP Relaxations for $I_{3322}$}
\label{Sec:SDP:Lowest}

Here, we will provide the explicit form for the matrices $F_k$ and
constants $b_k'$ that define the SDP used in the lowest level
relaxation for finding an upper bound on $\omega^*(I_{3322})$. Note
that as with the CHSH case, in the lowest level relaxation, we shall
choose $s_{ij}$ in Eq.~\eqref{Eq:qSOS} as multiples of $\id$. To this
end, we will write the SDP in the inequality form as
\begin{subequations}
\begin{gather}
    \text{minimize \ \ \ \ }\quad b_\nu'\,\nu+\sum_{k=1}^6
    b_k'\,x_k\label{Eq:SDP:1stLevel:Obj}\\
    \text{subject to \ \ }\Gamma=F_0+\nu F_\nu+\sum_{k=1}^6 x_k F_k \geq 0.
    \label{Eq:SDP:1stLevel:Cst}
\end{gather}\end{subequations}
In particular, we will set
\begin{gather*}
    F_0=\frac{1}{2}\left(\begin{array}{rrrrrrr}
    0 & 1 & 0 & 0 & 2 & 1 & 0\\
    1 & 0 & 0 & 0 & -1 & -1 & -1\\
    0 & 0 & 0 & 0 & -1& -1& 1\\
    0 & 0 & 0 & 0& -1& 1& 0\\
    2 & -1& -1& -1& 0 & 0 & 0\\
    1 & -1& -1& 1& 0 & 0 & 0\\
    0 & -1& 1& 0 & 0 & 0 & 0\\
    \end{array}\right),\quad
    F_\nu=\left(\begin{array}{rrrrrrr}
    1 & 0 & 0 & 0 & 0 & 0 & 0\\
    0 & 0 & 0 & 0 & 0 & 0 & 0\\
    0 & 0 & 0 & 0 & 0 & 0 & 0\\
    0 & 0 & 0 & 0 & 0 & 0 & 0\\
    0 & 0 & 0 & 0 & 0 & 0 & 0\\
    0 & 0 & 0 & 0 & 0 & 0 & 0\\
    0 & 0 & 0 & 0 & 0 & 0 & 0\\
    \end{array}\right),
\end{gather*}
and for $k=1,2,\ldots,6$, we shall choose $F_k$ such that its only
nonzero entries are $[F_k]_{1,k+1}=[F_k]_{k+1,1}=1/2$ and
$[F_k]_{k+1,k+1}=-1$. Correspondingly, $b_k'$ is chosen such that
$b_\gamma'=1$ and $b_k'=0$ for $k=1,2,\ldots,6$. With this choice of
$F_k$, $b_k'$ and with $z$ given by Eq.~\eqref{Eq:z:1st:3322}, it is
easy to see that the matrix inequality constraint,
Eq.~\eqref{Eq:SDP:1stLevel:Cst}, ensures that
\begin{equation*}
    z^\dagger\Gamma\,z=z^\dagger \left(F_0+\nu F_\nu+\sum_{k=1}^6 x_k F_k\right) z
    = -\mB_{3322}+\nu\,\id +\sum_{k=1}^6 x_k p_k = \text{SOS}
\end{equation*}
where $p_k$'s are defined in Eq.~\eqref{Eq:p_i:3322} and the last
equality follows from the positive semidefiniteness of $\Gamma$.

\subsection{Implementing Higher Level SDP Relaxations for $I_{3322}$}
\label{Sec:SDP:Higher}

In what follows, we will give a sketch of how the level 2 relaxation
for $I_{3322}$ inequality with $z$ given by Eq.~\eqref{Eq:z:1st+AB}
can be implemented as an SDP in the inequality form,
Eq.~\eqref{Eq:SDP:Vector}. Specifically, we want to write
Eq.~\eqref{Eq:SDP:abstract} as:
\begin{align}
    \text{minimize\ \ \ } &\qquad\qquad \nu,\nonumber\\
    \text{subject to\ \ } &F_0 + \nu F_\nu+\sum_k x_k F_k\ge 0,
    \label{Eq:SDP:SOS}
\end{align}
where, as with the lowest level relaxation, $F_0$ and $F_\nu$ are
real and symmetric matrices chosen such that
\begin{equation}\label{Eq:F0Fnu}
    z^\dagger\,F_0\,z=-\mB_{3322}, \quad z^\dagger \,F_\nu\,z=\id.
\end{equation}
Hereafter, we will focus on writing the second sum in
Eq.~\eqref{Eq:qSOS} as $\sum_k x_k z^\dagger F_k\,z$ for some
appropriate choice of Hermitian matrix $F_k$ where $x_k$ is some
variable to be optimized. As opposed to the lowest level relaxation,
the most general second level relaxation would require that each
$s_{ij}$ in Eq.~\eqref{Eq:qSOS} is a polynomial of degree at most 1.
Let $s_{ij}=\sum_k \lambda_{ijk}M_k$ where $M_k$ is the $k$-entry of
the vector $\mu=(\id,A^a_1,A^a_2,A^a_3,B^b_1,B^b_2,B^b_3)^\dagger$
which consists of all degree 1 or lower monomials that can be found
in $z$. For a fixed $i$ and $j$, we thus have
\begin{equation}\label{Eq:s_ij}
    s_{ij}^\dagger\,p_i\,s_{ij}
    =\sum_{k,l} M_k^\dagger\left(\lambda_{ijk}^*\lambda_{ijl}\right)p_i\,M_l
    =\mu^\dagger\Lambda_{ij}\, p_i\mu,
\end{equation}
where here $\lambda_{ijk}^*$ is the complex conjugate of
$\lambda_{ijk}$, $p_i\mu$ is a vector formed by multiplying each
entry of $\mu$ by $p_i$ and $\Lambda_{ij}$ is a $7\times7$ matrix
with its $(k,l)$-entry given by $\lambda_{ijk}^*\lambda_{ijl}$.
Clearly, as it is, $\Lambda_{ij}$ is a rank 1 but otherwise arbitrary
positive semidefinite matrix. Analogously, we see that if we further
perform a sum over $j$ in Eq.~\eqref{Eq:s_ij}, then we may write
$\sum_j s_{ij}^\dagger\,p_i\,s_{ij}=\mu^\dagger\Lambda_{i}\, p_i\mu$
where $\Lambda_i=\sum_j \Lambda_{ij}$ is now an arbitrary positive
semidefinite matrix. Moreover, the requirement of $\Lambda_i$ being
positive semidefinite can also be removed if we recall the fact that
in the case of $I_{3322}$ inequality, if $p_i$ is in $\cP$, so is
$-p_i$. Then, what remains to be done is to express $\sum_j
s_{ij}^\dagger\,p_i\,s_{ij}$ and hence $\mu^\dagger\Lambda_{i}\,
p_i\mu$ in the form $z^\dagger\Omega_i z$ for some Hermitian
$\Omega_i$. Evidently, the entries in $\Lambda_i$ will be related to
the entries in $\Omega_i$ linearly. For example, since
\begin{align*}
    [\Lambda_1]_{6,5}M_6^\dagger p_1 M_5
    =&[\Lambda_1]_{6,5}(B^b_2)^\dagger \left[(A^a_1)^2-(A^a_1)\right](B^b_1)\\
    =&[\Lambda_1]_{6,5}\left[(A^a_1B^b_2)^\dagger (A^a_1B^b_1)
    -(B^b_2)^\dagger (A^a_1B^b_1)\right],
\end{align*}
we may make the following identification
$[\Omega_1]_{9,8}=-[\Omega_1]_{6,8}=[\Lambda_1]_{6,5}$ and so on.
Each independent entry of $\Lambda_i$ therefore corresponds to an
independent optimization variable $x_k$ and some Hermitian matrix
$F_k$ via $\Omega_i$. In the example above, the $F_k$ corresponding
to (the real part of) $[\Lambda_1]_{6,5}$ would be zero everywhere
except for its $(6,8)$, $(9,8)$, $(8,6)$ and $(8,9)$ entry, which
reads as $-1,1,-1,1$ respectively. More intuitively, each of these
$F_k$'s gives rise to some polynomial identities such that $z^\dagger
F_k z=0$ whenever the constraints are satisfied. Putting everything
together, we see that the search for a positive semidefinite $\Gamma$
such that $\nu\,\id-\cB_{3322}-\sum_{ij} s_{ij}^\dagger\, p_i
s_{ij}=z^\dagger\Gamma z$ can also be written as the search for a
positive semidefinite $\Gamma$ such that
\begin{equation}
    \nu\,\id-\cB_{3322}-\sum_k x_k z^\dagger F_k z =z^\dagger\Gamma z,
\end{equation}
for some appropriate choice of $F_k$. Comparing this with
Eq.~\eqref{Eq:SDP:SOS} and Eq.~\eqref{Eq:F0Fnu}, we see that
evidently any higher order relaxation in our hierarchy can also be
implemented as an SDP.

\subsection{Yao's inequality}\label{App:Yao}

Here, we note that for $\{A_i,B_j,C_k\}_{i,j,k=1,2,3}$ satisfying the
commutation relations $[A_i,B_j]=0$, $[A_i,C_k]=0$, $[B_j,C_k]=0$ and
for $\Gamma$ given by Eq.~\eqref{Eq:Gamma:Blks}, we have
\begin{align}
    3\sqrt{3}\,\id - \mB_{\Yao} &= z^\dagger \Gamma z
    +\sum_{i,j,k} \alpha_{ijk}\left(\id - t_{ijk}^\dagger t_{ijk}\right)\nonumber\\
    &+\frac{1}{12\sqrt{3}}\sum_{i,j,k}\,^{'} \left(f^{(A)}_{ijk} +f^{(B)}_{ijk}
    +f^{(C)}_{ijk} + f^{(A)\,\dagger}_{ijk} +f^{(B)\,\dagger}_{ijk}+f^{(C)\,\dagger}_{ijk}\right),
    \label{Eq:Yao:SOS}
\end{align}
where $t_{ijk}:=A_iB_jC_k$,
\begin{gather*}
    \alpha_{ijk}=\left\{ \begin{array}{c@{\quad:\quad}l}
        \frac{1}{2\sqrt{3}}  & i\neq j\neq k,\\
        0 & i=j=k,\\
        \frac{1}{12\sqrt{3}} & \text{otherwise.}
        \end{array} \right.\quad
    f^{(l)}_{ijk}=\left\{ \begin{array}{c@{\quad:\quad}l}
        2t_{ijk}^\dagger t_{ikj} - t_{jjk}^\dagger t_{jkj}
        - t_{kjk}^\dagger t_{kkj}  & l=A,\\
        2t_{ijk}^\dagger t_{kji} - t_{iik}^\dagger t_{kii}
        - t_{ikk}^\dagger t_{kki}  & l=B,\\
        2t_{ijk}^\dagger t_{jik} - t_{ijj}^\dagger t_{jij}
        - t_{iji}^\dagger t_{jii}  & l=C,
        \end{array} \right.\quad
\end{gather*}
and the second sum $\sum\,^{'}$ is over all possible $i,j,k$ such
that $i\neq j\neq k$. In contrast with
Eq.~\eqref{Eq:Yao:SOSwoConstraints}, the above equality holds even if
none of the constraints $A_i^2=B_j^2=C_k^2=\id$ are satisfied.
Moreover, the above equality can also be cast in the form of
Eq.~\eqref{Eq:qSOS} with the help of identities such as
\begin{gather*}
    \id-t_{ijk}^\dagger t_{ijk}=
    p^{(C)}_k + C_k^\dagger\,p^{(B)}_jC_k + g_{jk}^\dagger\,p^{(A)}_i\,g_{jk},
\end{gather*}
and
\begin{align*}
    f^{(A)}_{ijk}+f^{(A)\,\dagger}_{ijk}
    &=2g_{jk}^\dagger\,p^{(A)}_ig_{jk} + 2g_{kj}^\dagger\,p^{(A)}_ig_{kj}
     - g_{jk}^\dagger\,\left(p^{(A)}_j+p^{(A)}_k\right)g_{jk}
     - g_{kj}^\dagger\,\left(p^{(A)}_j+p^{(A)}_k\right)g_{kj}\\
     &+ \left(g_{jk}+g_{kj}\right)^\dagger p^{(A)}_j\left(g_{jk}+g_{kj}\right)
     + \left(g_{jk}+g_{kj}\right)^\dagger p^{(A)}_k\left(g_{jk}+g_{kj}\right)\\
     &-2\left(g_{jk}+g_{kj}\right)^\dagger p^{(A)}_i\left(g_{jk}+g_{kj}\right),
\end{align*}
where $g_{jk}:=B_jC_k$.

\end{document}